%% file: arXiv_v1.tex
\documentclass[11pt]{article}

\usepackage{amsmath}
\usepackage{amssymb}
\usepackage{amsthm}
\usepackage{amsfonts}
\usepackage[pagebackref]{hyperref}
\hypersetup{
    colorlinks=true, 
    linkcolor=blue, 
    citecolor=blue, 
    urlcolor=blue 
}
\usepackage{microtype}

\usepackage{accents}
\usepackage[boxed,titlenumbered]{algorithm2e}
\usepackage[toc,page]{appendix}
\usepackage{array}
\usepackage[english]{babel}
\usepackage{braket}
\usepackage{booktabs}
\usepackage{bigdelim}
\usepackage{cancel}
\usepackage[shortlabels]{enumitem}
\usepackage{float}
\usepackage[bottom]{footmisc}
\usepackage[utf8]{inputenc}
\usepackage[margin=1in]{geometry}
\usepackage{graphicx}
\usepackage[capitalize,nameinlink]{cleveref} 
\usepackage{mathtools}
\usepackage{multirow}
\usepackage{physics}
\usepackage{tabularx}
\usepackage[dvipsnames]{xcolor}
\usepackage{tikz}
\usepackage[textsize=footnotesize]{todonotes}
\usepackage{qcircuit}
\usetikzlibrary{positioning}
\usetikzlibrary{decorations.pathreplacing}
\usetikzlibrary{cd}
\usetikzlibrary{arrows.meta}
\usetikzlibrary{calc}
\usepackage{bbold}

\usepackage{tikzsymbols}
\usepackage{complexity}

\usepackage{xspace}

\renewcommand{\backref}[1]{}

\renewcommand{\backrefalt}[4]{%
\ifcase #1 %
\or
[p.\ #2]%
\else
[pp.\ #2]%
\fi}

\makeatletter
\renewcommand{\paragraph}{%
 \@startsection{paragraph}{4}%
 {\z@}{2.25ex \@plus .5ex \@minus .3ex}{-1em}%
 {\normalfont\normalsize\bfseries}%
}
\makeatother

\makeatletter
\newcommand{\para}{%
 \@startsection{paragraph}{4}%
 {\z@}{2ex \@plus 3.3ex \@minus .2ex}{-1em}%
 {\normalfont\normalsize\bfseries}%
}
\makeatother

\DeclarePairedDelimiter{\ceil}{\lceil}{\rceil}

\newcolumntype{P}[1]{>{\centering\arraybackslash}p{#1}}

\makeatletter
\newtheorem*{rep@theorem}{\rep@title}
\newcommand{\newreptheorem}[2]{%
\newenvironment{rep#1}[1]{%
 \def\rep@title{#2 \ref{##1} (formal)}%
 \begin{rep@theorem}}%
 {\end{rep@theorem}}}
\makeatother

\newtheorem{lem}{Lemma}
\newtheorem{thm}[lem]{Theorem}
\newtheorem*{thm*}{Theorem}
\newtheorem{prop}[lem]{Proposition}

\newtheorem{cor}[lem]{Corollary}
\newtheorem{defn}[lem]{Definition}

\theoremstyle{definition}

\newtheorem{problem}{Problem}

\newreptheorem{thm}{Theorem}
\newreptheorem{lem}{Lemma}
\newreptheorem{cor}{Corollary}
\newreptheorem{defn}{Definition}
\newreptheorem{claim}{Claim}
\newreptheorem{fact}{Fact}
\newreptheorem{con}{Conjecture}

\newcommand\cE{{\cal E}}

\newcommand\eps{{\varepsilon}}

\newcommand{\nc}{\newcommand}
\nc{\rnc}{\renewcommand}
\nc\benum{\begin{enumerate}}
\nc\eenum{\end{enumerate}}

\nc\ot{\otimes}
\nc\bit{\begin{itemize}}
\nc\eit{\end{itemize}}

\def\be#1\ee{\begin{equation}#1\end{equation}}
\def\ba#1\ea{\begin{align}#1\end{align}}
\def\bas#1\eas{\begin{align}#1\end{align}}
\def\bpm#1\epm{\begin{pmatrix}#1\end{pmatrix}}

\def\be#1\ee{\begin{equation}#1\end{equation}}
\def\bea#1\eea{\begin{eqnarray}#1\end{eqnarray}}
\def\beas#1\eeas{\begin{eqnarray*}#1\end{eqnarray*}}
\def\ba#1\ea{\begin{align}#1\end{align}}
\def\bas#1\eas{\begin{align}#1\end{align}}
\def\bpm#1\epm{\begin{pmatrix}#1\end{pmatrix}}
\def\bbm#1\ebm{\begin{bmatrix}#1\end{bmatrix}}
\def\bbml#1\ebml{\begin{bmatrix*}[l]#1\end{bmatrix*}}

\newcommand{\inputlc}[1]{LC(#1)}
\newcommand{\outputlc}[1]{LC(#1)}

\mathchardef\mhyphen="2D

\newclass{\QNCcat}{\QNC^{0}/\Cat}
\newcommand{\NCz}{$\NC^0$\xspace}

\newcommand{\ACz}{$\AC^0$\xspace}

\newcommand{\rpoly}{\mathsf{rpoly}}

\renewcommand{\>}{\rangle}
\newcommand{\B}{\{0,1\}}
\newcommand{\PHP}{\mathrm{PHP}}
\newcommand{\RPHP}{\mathrm{RPHP}}

\newcommand{\GRPHP}{\mathrm{Grid{\mhyphen}RPHP}}

\newcommand{\PGRPHP}{\mathrm{Parallel\ Grid{\mhyphen}RPHP}}
\newcommand{\PPHP}{\mathrm{Parallel{\mhyphen}PHP}}

\newcommand{\DDHLF}{\textup{2D HLF}\xspace}
\newcommand{\HLF}{\textup{HLF}\xspace}
\renewcommand{\PBP}{\mathrm{PBP}}

\newcommand{\PPBP}{\mathrm{Parallel\ PBP}}

\title{Exponential separation between shallow quantum circuits\\ and unbounded fan-in shallow classical circuits}

\author{Adam Bene Watts\thanks{Massachusetts Institute of Technology. \texttt{abenewat@mit.edu}. Supported by NSF grant CCF-1729369.
} \and
Robin Kothari\thanks{Microsoft Quantum and Microsoft Research. \texttt{robin.kothari@microsoft.com}} \and
Luke Schaeffer\thanks{Massachusetts Institute of Technology. \texttt{lrs@mit.edu}. Part of this work was done while the author was an intern in the Quantum Architectures and Computation group (QuArC), Microsoft Research.} \and
Avishay Tal\thanks{Stanford University. \texttt{avishay.tal@gmail.com}. This work was done in part while the author was visiting the Simons Institute for the Theory of Computing. Partially supported by a Motwani Postdoctoral Fellowship and by NSF grant CCF-1763311.}}

\begin{document}
\date{}
\maketitle


\begin{abstract}
Recently, Bravyi, Gosset, and K\"{o}nig (Science, 2018) exhibited a search problem called the 2D Hidden Linear Function ($\DDHLF$) problem that can be solved exactly by a constant-depth quantum circuit using bounded fan-in gates (or $\QNC^0$ circuits), but cannot be solved by any constant-depth classical circuit using bounded fan-in AND, OR, and NOT gates (or $\NC^0$ circuits). In other words, they exhibited a search problem in $\QNC^0$ that is not in $\NC^0$. 

We strengthen their result by proving that the 2D HLF problem is not contained in $\AC^0$, the class of classical, polynomial-size, constant-depth circuits over the gate set of \emph{unbounded} fan-in AND and OR gates, and NOT gates. 
We also supplement this worst-case lower bound with an average-case result:
There exists a simple distribution under which any $\AC^0$ circuit (even of nearly exponential size) has exponentially small correlation with the 2D HLF problem. 
Our results are shown by constructing a new problem in $\QNC^0$, which we call the Relaxed Parity Halving Problem, which is easier to work with. We prove  our $\AC^0$ lower bounds for this problem, and then show that it reduces to the 2D HLF problem. 

As a step towards even stronger lower bounds, we present a search problem that we call the Parity Bending Problem, which is in $\QNC^0/\qpoly$ ($\QNC^0$ circuits that are allowed to start with a quantum state of their choice that is independent of the input), but is not even in $\AC^0[2]$ (the class $\AC^0$ with unbounded fan-in XOR gates).

All the quantum circuits in our paper are simple, and the main difficulty lies in proving the classical lower bounds. 
For this we employ a host of techniques, including a refinement of H{\aa}stad's switching lemmas for multi-output circuits that may be of independent interest, the Razborov-Smolensky $\AC^0[2]$ lower bound, Vazirani's XOR lemma, and lower bounds for non-local games.

\end{abstract}

\pagenumbering{gobble}
\clearpage

\tableofcontents
\clearpage
\pagenumbering{arabic}

\section{Introduction}

One of the basic goals of quantum computing research is to identify problems that quantum computers can solve more efficiently than classical computers. 
We now know several such problems, such as the integer factorization problem, which we believe can be solved exponentially faster on a quantum computer~\cite{Sho97}.
However, running this algorithm requires a large general-purpose quantum computer, which we do not yet have.
Hence it is interesting to find examples of quantum speedup using weaker models of quantum computation, such as models with limited space or time, limited gate sets, or limited geometry of interactions.

\para{Shallow quantum circuits.} One such model of quantum computation that has been studied for over 20 years is the class of shallow or constant-depth quantum circuits~\cite{MN02,Moo99,GHMP02,TD04,FGHZ03,HS05,FFGHZ06,TT11}. 
Such circuits may be viewed as parallel quantum computers with a constant running time bound.
Several variations on this theme have been studied (see \cite{BGH07} for a survey of older results), and in recent years there has been a resurgence of interest~\cite{TT18,BHSRE18,BGK18,CSV18,LG18} in constant-depth quantum circuits, for at least two reasons.

First, shallow quantum circuits are well motivated from a practical perspective, as we might actually be able to implement such circuits on near-term quantum computers!
In the current era of Noisy Intermediate-Scale Quantum (NISQ) computers, due to high error rates of quantum gates, we are limited to running quantum algorithms for a short amount of time before errors accumulate and noise overwhelms the signal.
Hence we seek interesting problems that can still be implemented by limited quantum hardware.

Second, constant-depth circuits (either classical or quantum) are very interesting to theoretical computer scientists, as it is possible to prove unconditional impossibility results about constant-depth circuits. 
For example, while we strongly believe that the factoring problem mentioned above requires exponential time on a classical computer, we cannot prove this.
On the other hand, many of the early successes of complexity theory involved exhibiting explicit functions that could not be computed by constant-depth classical circuits~\cite{Ajtai83, FSS84,Yao85, Has86}. 
Indeed, constant-depth circuits remain the frontier of circuit lower bounds and an active area of research in classical complexity theory today~\cite{Wil14,MW18}.

This motivates the search for problems that can be solved by constant-depth quantum circuits, while being hard for constant-depth (or even more powerful) classical circuits. 

\para{Prior work.} 
While there has been prior work on establishing the power of shallow quantum circuits assuming complexity theoretic conjectures~\cite{TD04,BHSRE18}, this work is not directly related to our work as we prove unconditional lower bounds.

In this realm the most relevant result is the recent exciting result of Bravyi, Gosset, and K\"{o}nig~\cite{BGK18}, who defined a search or relational problem\footnote{A search or relational problem can have many valid outputs for a given input, unlike a function problem that has exactly one valid output. A decision problem is a function problem with a $1$-bit output.} 
called the 2D Hidden Linear Function (2D HLF) problem.
(We define this problem in \Cref{sec:HLF}.)
The 2D HLF problem can be solved by a constant-depth quantum circuit that uses bounded fan-in quantum gates. 
Indeed, the quantum circuit solving 2D HLF can be implemented on a 2-dimensional grid of qubits with spatially local quantum gates.

Furthermore, Bravyi, Gosset, and K\"{o}nig~\cite{BGK18} show that the \DDHLF problem cannot be solved by any constant-depth classical circuit using unbounded fan-out and bounded fan-in gates. 
Their lower bound even holds when the classical circuit is allowed to sample from an arbitrary probability distribution on polynomially many bits that does not depend on the input. (In complexity theory, this resource is called ``randomized advice.'')
More formally, the class of classical circuits of polynomial-size, constant-depth, unbounded fan-out, and bounded fan-in gates is called $\NC^0$.\footnote{In this paper, we will employ a common abuse of notation and use class names like $\NC^0$ and $\AC^0$ to generally talk about a type of circuit, as opposed to decision problems solved by such circuits. Hence, for example, we speak of ``decision problems in $\AC^0$'' and ``search problems in $\AC^0$'' although formally $\AC^0$ would be the class of decision problems solved by such circuits, and $\mathsf{FAC}^0$ would be the class of search problems solved by such circuits.}
An $\NC^0$ circuit with the additional ability to sample from any probability distribution on polynomially many bits that is independent of the input, but that can depend on the input size, is called an $\NC^0/\rpoly$ circuit. The class of polynomial-size, constant-depth quantum circuits with bounded fan-in gates is called $\QNC^0$.
Note that because quantum gates have the same number of inputs and outputs, $\QNC^0$ circuits also have bounded fan-out, unlike classical $\NC^0$ circuits, which have unbounded fan-out.

With this notation, we can now summarize the Bravyi et al. result as follows~\cite{BGK18}.

\begin{thm*}[Bravyi, Gosset, and K\"{o}nig]
The \DDHLF problem can be solved exactly by a $\QNC^0$ circuit on a 2D grid, but no $\NC^0/\rpoly$ circuit can solve the problem with probability greater than $7/8$ on every input.
\end{thm*}

The fact that the separating problem in \cite{BGK18} is a search problem and not a function (or decision) problem is unavoidable, since any function in $\QNC^0$ has output bits that only depend on a constant number of input bits, due to the bounded fan-in gates, and hence such a function would also be in $\NC^0$.

This result was also recently improved by Coudron, Stark, and Vidick~\cite{CSV18}, and (independently) Le Gall~\cite{LG18}, who extended the lower bound to an average-case lower bound. As opposed to saying that no $\NC^0$ circuit can solve the problem on \emph{all} inputs, an average-case hardness result says that no $\NC^0$ circuit can solve the problem even on some fraction of the inputs.\footnote{Note that  \cite[Appendix C.3]{BGK18} already shows mild average-case hardness for this problem.}  
These results show that no $\NC^0$ circuit can solve the problem with input size $n$ on an $\exp(-n^\alpha)$ fraction of the inputs for some $\alpha>0$.

\para{Main result.} In this work, we strengthen these results and prove a strong average-case lower bound for the 2D HLF problem against the class $\AC^0$.
$\AC^0$ is a natural and well-studied class that generalizes $\NC^0$ by allowing the circuit to use unbounded fan-in AND and OR gates. 
Note that $\NC^0 \subsetneq \AC^0$ because $\AC^0$ can compute functions that depend on all bits, such as the logical OR of all its inputs, whereas $\NC^0$ cannot. Our main result is the following.

\begin{thm}[\DDHLF]\label{thm:main}
The \DDHLF problem on $n$ bits cannot be solved by an $\AC^0$ circuit of depth $d$ and size at most $\exp(n^{1/10d})$.
Furthermore, there exists an (efficiently sampleable) 
input distribution on which any $\AC^0$ circuit (or  $\AC^0/\rpoly$ circuit) of depth $d$ and size at most $\exp(n^{1/10d})$ only solves the \DDHLF problem with probability at most $\exp(-{n^\alpha})$ for some $\alpha>0$.
\end{thm}

Thus our result proves a separation against a larger complexity class and implies the worst-case lower bound of Bravyi, Gosset, and K\"{o}nig~\cite{BGK18}. It also implies the average-case lower bounds of  Coudron, Stark, and Vidick~\cite{CSV18} and Le Gall~\cite{LG18}.

\subsection{High-level overview of the main result}

We now describe the problems we study en route to proving \Cref{thm:main} and give a high-level overview of the proof.

\Cref{thm:main} is proved via a sequence of increasingly stronger results. We first introduce a problem we call the Parity Halving Problem (PHP). PHP is not in $\QNC^0$, but it can be solved exactly by a $\QNC^0/\qpoly$ circuit, which is a $\QNC^0$ circuit with quantum advice. Similar to randomized advice, a circuit class with quantum advice is allowed to start with any polynomial-size quantum state that is independent of the input, but can depend on the input length. 
For the Parity Halving Problem (and other problems introduced later), the quantum advice state is a very simple state called the cat state, which we denote by $|\Cat_n\> := \frac{1}{\sqrt{2}} (|0^n\>+|1^n\>)$.
We denote the subclass of $\QNC^0/\qpoly$ where the advice state is the cat state $\QNCcat$.

Here's a bird's eye view of our proof: Our first result establishes that $\PHP$ is in $\QNCcat$, but any nearly exponential-size $\AC^0$ circuit only solves the problem with probability exponentially close to $1/2$. Next we define a new problem called the Relaxed Parity Halving Problem on a grid ($\GRPHP$), which is indeed in $\QNC^0$, but any nearly exponential-size $\AC^0$ circuit only solves the problem with probability exponentially close to $1/2$. We then define  parallel versions of these two problems, which we call $\PPHP$ and  $\PGRPHP$. 
We show that $\PPHP \in \QNC^0/\qpoly$ and $\PGRPHP \in \QNC^0$, but any nearly exponential-size $\AC^0$ circuit only solves these problems with exponentially small probability. 
Finally we show that $\PGRPHP$ can be reduced to \DDHLF, and hence our lower bound applies to \DDHLF as well. We now describe these problems and our proof techniques in more detail.

\para{Parity Halving Problem.} In the Parity Halving Problem on $n$ bits, which we denote by $\PHP_n$, we are given an input string $x \in \B^n$ promised to have even parity: i.e., the Hamming weight of $x$, denoted $|x|$, satisfies $|x| \equiv 0 \pmod 2$. The goal is to output a string $y \in \B^n$ that satisfies
\begin{equation}
    |y| \equiv |x|/2 \pmod 2.
\end{equation} 
In other words, the output string's Hamming weight (mod 2) is half of that of the input string.
Note that $|x|/2$ is well defined above because $|x|$ is promised to be even.
An alternate way of expressing this condition is that $|y|\equiv 0 \pmod 2$ if $|x|\equiv 0 \pmod 4$ and $|y|\equiv 1 \pmod 2$ if $|x|\equiv 2 \pmod 4$. 

We show in \Cref{sec:PHP} that $\PHP$ can be solved with certainty on every input by a simple depth-2 $\QNCcat$ circuit. A quantum circuit solving $\PHP_3$ is shown in \Cref{fig:php}. The circuit has one layer of controlled phase gates followed by Hadamard gates on the output qubits, followed by measurement.

\begin{figure}[b]
\vspace{-0.75em}
\begin{center}
\input{shortdepth.tex}
\end{center}
\vspace{-1em}
\caption{Quantum circuit for the Parity Halving Problem on $3$ bits, $\PHP_3$.}
\label{fig:php}
\end{figure}
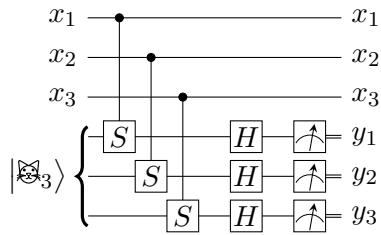

Although the problem is easy for constant-depth quantum circuits, we show that even an exponential-size $\AC^0/\rpoly$ circuit cannot solve the problem on the uniform distribution (over valid inputs) with probability considerably better than $1/2$, which is trivially achieved by the circuit that outputs the all-zeros string on all inputs.

\begin{thm}[PHP]\label{thm:PHP}
The Parity Halving Problem $(\PHP_n)$ can be solved exactly by a $\QNCcat$ circuit. 
But on the uniform distribution over all valid inputs (even parity strings), any $\AC^0/\rpoly$ circuit of depth $d$ and size at most $\mathrm{exp}\bigl(n^{\frac{1}{10d}}\bigr)$ only solves the problem with probability $\frac{1}{2}+\exp(-n^{\alpha})$ for some $\alpha>0$. 
\end{thm}

Note that the parameters of the \ACz lower bound in this theorem are essentially optimal, since the parity function on $n$ bits can be computed by a depth-$d$ \ACz circuit of size $\mathrm{exp}\bigl(n^{\frac{1}{d-1}}\bigr)$ \cite[Theorem 2.2]{Has86}. Once we can compute the parity of the input bits, it is easy to solve $\PHP$.

Since the quantum circuit for $\PHP$ is simple, it is clear that the difficult part of \Cref{thm:PHP} is  the $\AC^0$ lower bound. One reason for this difficulty is that if we allowed the output string $y$ in $\PHP$ to be of quadratic size, then there is a simple depth-$1$ $\NC^0$ circuit that solves this problem! The circuit simply computes the AND of every pair of input bits and outputs this string of size $\binom{n}{2}$. 
A simple calculation shows that the Hamming weight of this string will be $\binom{|x|}{2}$, which satisfies the conditions of the problem.
Hence to prove the $\AC^0$ lower bound, the technique used has to be sensitive to the output size of the problem. 
However, traditional \ACz lower bound techniques were developed for decision problems, and do not explicitly take the output size of the problem into account. 
Hence we modify some known techniques and establish the this lower bound in three steps.

First, we use H{\aa}stad's switching lemmas~\cite{Has86,Hastad14}, or more precisely a recent refinement of it due to Rossman~\cite{Rossman2017}. However, directly using the result of Rossman off the shelf gives us a weaker result than \Cref{thm:PHP}; we are only able to establish the theorem with a quasi-polynomially small correlation instead of the exponentially small correlation in \Cref{thm:PHP}.
To obtain the result we want, we refine Rossman's result to work better for multi-output functions (\Cref{cor:Rossman}). 
This result is quite technical, but the conceptual ideas already appear in the works of H{\aa}stad \cite{Hastad14} and Rossman~\cite{Rossman2017}. Applying this switching lemma reduces the problem of proving an average-case $\AC^0$ lower bound to that of showing an average-case $\NC^0$ lower bound for a modified version of the Parity Halving Problem.
This modified version of the problem is similar to PHP, except it has $n$ inputs and slightly more (say, $n^{1.01}$) outputs.

Our second step is to use a combinatorial argument to reduce this question to showing an average-case lower bound against $\NC^0$ circuits with locality $1$ (i.e., where each output only depends on a single input) for a further modified version of PHP. 

The third and final step is to show that $\NC^0$ circuits with locality $1$ cannot solve this modifed PHP  on a random input. We prove this by generalizing known lower bounds in the literature on quantum non-local games. Specifically we generalize lower bounds present in the work of Mermin \cite{Mermin90}, and Brassard, Broadbent, and Tapp \cite{BBT05}.

This proof is presented in \Cref{sec:PHP}. We first prove the lower bound against $\NC^0$ circuits of locality 1 in \Cref{sec:NClocalityone}, then show the lower bound against general $\NC^0$ circuits in \Cref{sec:NC}, and finally introduce the switching lemma and conclude the proof of \Cref{thm:PHP} in \Cref{sec:AC}. The switching lemma itself is proved in \Cref{app:Rossman}.

Now \Cref{thm:PHP} is weaker than what we want (\Cref{thm:main}) in two ways. Aside from the fact that the lower bound is for a problem different from the \DDHLF problem, the problem in \Cref{thm:PHP} is in $\QNCcat$ and not $\QNC^0$, and the correlation lower bound is close to $1/2$ instead of being exponentially small. We now tackle the first problem and get rid of the cat state.

\para{Relaxed Parity Halving Problem.} 
Since the cat state cannot be constructed in $\QNC^0$ (proved in \Cref{thm:nocat}), 
we have to modify the Parity Halving Problem  to get by without a cat state.

Although we cannot create the cat state in $\QNC^0$, we can construct a state we call a ``poor man's cat state,'' which is the state 
\begin{equation}
    \frac{1}{\sqrt{2}}\bigl(|z\>+|\bar{z}\>\bigr),
\end{equation}
where $z \in \B^n$ is a bit string and $\bar{z}$ denotes its complement. 
When $z=0^n$, this is indeed the cat state, but in general this is some entangled state that can be converted to the cat state by applying the $X$  gate to some subset of the qubits. 

Interestingly, we can create a poor man's cat state in $\QNC^0$ for a uniformly random $z$. 
Here is one simple construction. First arrange the $n$ qubits on a line and set them all to be in the $|+\>$ state. Then in a separate set of $n-1$ qubits, compute the pairwise parities of adjacent qubits. In other words, we store the parity of qubit 1 and 2, 2 and 3, 3 and 4, and so on until qubit $n-1$ and $n$. 
And then we measure these $n-1$ qubits, and denote the measurement outcomes $d_1, \ldots, d_{n-1}$, which we will call the ``difference string.'' It is easy to verify that if all the $d_i=0$, then the resulting state is indeed the cat state. 
In general, for any $d\in\B^{n-1}$, the resulting state is a poor man's cat state, and $z$ can be determined from the string $d$ up to the symmetry between $z$ and $\bar{z}$.
Since $z$ and $\bar{z}$ are symmetric in the definition of the poor man's cat state, let us choose the convention that $z_1 = 0$. 
Now we can determine the remaining $z_i$ from $d$ using the fact that $d_1 = z_1 \oplus z_2 = z_2$, $d_2 = z_2 \oplus z_3$, and so on until $d_{n-1} = z_{n-1} \oplus z_n$. 
Note that because of this construction, some $z_i$ depend on many bits of $d_i$. 
For example, $z_{n}$ is the parity of all the bits in $d$. 

This construction of the poor man's cat state easily generalizes to graphs other than the 1D line. We could place the $n$ qubits on a balanced binary tree and measure the parity of all adjacent qubits, and hence get one $d_i$ for every edge in the tree. If we call the root node $z_1 = 0$, then the value of any $z_i$ will be the parity of all $d_i$ on edges between vertex $i$ and the root. In this case each $z_i$ depends on at most $\log n$ bits of $d_i$. Similarly, we can choose a 2D grid instead of a balanced binary tree, and set the top left qubit to be $z_1=0$. Then each $z_i$ will depend on at most $2\sqrt{n}$ bits of $d$. This grid construction is described more formally in \Cref{sec:PoorMansCatState}. 

Now there's an obvious strategy to try: Simply use a poor man's cat in our quantum circuit for $\PHP$ instead of using an actual cat state, and redefine the problem to match the output of this quantum circuit! So we simply run the circuit in \Cref{fig:php} on a poor man's cat state $\frac{1}{\sqrt{2}}\bigl(|z\>+|\bar{z}\>\bigr)$ and see what the quantum circuit outputs. 
Unfortunately the output depends on $z$, but the poor man's cat state has been destroyed by the circuit and we do not have a copy of $z$ around.
But we do still have the string $d$ from which it is possible to recover $z$, although this may not be computationally easy since a single bit of $z$ may depend on a large number of bits of $d$. 
More subtly, a single bit of $d$ may be involved in specifying many bits of $z$, which is also a problem for circuits without fan-out, such as $\QNC^0$ circuits.  
Instead of trying to recover $z$, we can just modify the problem to include $d$ as an output.
The problem will now have two outputs, one original output $y$, and a second output string $d$, which is the difference string of the $z$ in the poor man's cat state.
This is the Relaxed Parity Halving Problem, which is more formally defined in \Cref{sec:RelaxedParityHalving}.

More precisely, the Relaxed Parity Halving Problem, or $\RPHP$, depends on the choice of the underlying graph, and is well defined for any graph. 
We choose the 2D grid to get a problem that reduces to \DDHLF.\footnote{Picking the balanced binary tree would give better parameters, but qualitatively similar results. We choose the 2D grid so that our problem can be solved by a constant-depth quantum circuit acting on qubits laid out in 2D.}
We call this problem $\GRPHP$.

We show in \Cref{sec:RelaxedParityHalving} that $\GRPHP$ can be solved by the 2D $\QNC^0$ circuit we described, but even a nearly exponentially large $\AC^0$ circuit cannot solve the problem with probability significantly larger than $1/2$ on the uniform distribution over valid inputs.

\begin{thm}[$\GRPHP$]
\label{thm:GRPHP}
$\GRPHP_n$ can be solved exactly by a $\QNC^0$ circuit on a 2D grid.
But on the uniform distribution over all valid inputs (even parity strings), any $\AC^0$ circuit (or  $\AC^0/\rpoly$ circuit) of depth $d$ and size at most $\exp(n^{1/10d})$ can solve the problem with probability at most $\frac{1}{2}+\exp(-n^{\alpha})$ for some $\alpha>0$.
\end{thm}

Note that just like \Cref{thm:PHP}, the lower bound here is essentially optimal, since the parity function itself can be computed by a depth-$d$ \ACz circuit of size $\mathrm{exp}\bigl(n^{\frac{1}{d-1}}\bigr)$ \cite[Theorem 2.2]{Has86}. 
Our separation essentially works for any graph with sublinear diameter, such as the grid or the balanced binary tree, but not the 1D line. In fact, when the underlying graph is the 1D line, $\RPHP$ becomes easy to solve, even for $\NC^0$ circuits.%
\footnote{One can output $y = 0^n$ and $d_i = x_i$ for all $i\in\{1, \ldots, n-1\}$ to solve the Relaxed Parity Halving Problem on the 1D line.\label{footnote:line}}

We prove \Cref{thm:GRPHP} by showing a reduction from the Parity Halving Problem with input size $n$ and output size $O(n^{3/2})$ to $\GRPHP$. This version of PHP is indeed hard for $\AC^0$ circuits and this result follows from the work done in \Cref{sec:PHP}. This reduction and theorem are proved formally in \Cref{sec:RelaxedParityHalving}.

Now \Cref{thm:GRPHP} is still weaker than what we want (\Cref{thm:main}). The correlation lower bound is still close to $1/2$ and not exponentially small. 
We now fix this issue using a simple idea.

\para{Parallel Grid-RPHP.} Let $\PGRPHP$ be the problem where we are given many instances of $\GRPHP$ in parallel and are required to solve all of them correctly. 
For this problem the quantum circuit is obvious: Simply use the quantum circuit for $\GRPHP$ for each instance of the problem. Clearly if the quantum circuit solves each instance correctly, it solves all of them correctly. But since a classical circuit only solves an instance with some probability close to $1/2$, we expect that solving many copies of the problem gets much harder.

\begin{thm}[$\PGRPHP$]
\label{thm:PGRPHP}
$\PGRPHP_n$ can be solved exactly by a $\QNC^0$ circuit on a 2D grid.
But on the uniform distribution over all valid inputs (even parity strings for each instance of $\GRPHP$), any $\AC^0$ circuit (or  $\AC^0/\rpoly$ circuit) of depth $d$ and size at most $\exp(n^{1/10d})$ can solve the problem with probability at most  $\exp(-{n^\alpha})$ for some $\alpha>0$.
\end{thm}

As before, the difficult part of \Cref{thm:PGRPHP} is proving the classical lower bound.
While it seems intuitive that repeating the problem (in parallel)  several times reduces the success probability, similarly intuitive statements can be false or difficult to prove~\cite{Raz98}. More precisely, what we need is a direct product theorem, which also may not hold in some models of computation~\cite{Sha04}.

We consider the parallel version of the standard $\PHP$, denoted $\PPHP$, and reduce $\PPHP$ to $\PGRPHP$ as above.
We then establish a lower bound for $\PPHP$ by using Vazirani's XOR lemma~\cite{Vaz86}.
Vazirani's XOR lemma is an intuitive statement about how a probability distribution that is ``balanced'' in a certain sense must be close to the uniform distribution.
The implication for our problem is the following: To understand the probability that a circuit solves all the instances of $\PHP$ in $\PPHP$, or equivalently that it fails to solve $0$ instances, it is enough to understand the probability that it fails to solve an even number of instances. 
This task turns out to be similar to the original $\PHP$ with larger input and output size, but with some additional constraints on the input. 
The techniques we have developed allow us to upper bound this probability, and hence (using the XOR lemma) upper bound the probability that a circuit solves all instances correctly.

Now we are almost done, since \Cref{thm:PGRPHP} looks very similar to \Cref{thm:main}, except that the hardness is shown for $\PGRPHP$ and not the \DDHLF problem.

\para{Reduction to the Hidden Linear Function problem.} The final step of our program is carried out in  \Cref{sec:HLF}. First we show via a simple reduction in \Cref{thm:hlfreduction} that the Relaxed Parity Halving Problem (for any graph $G$) can be reduced to the Hidden Linear Function problem (not necessarily the \DDHLF). In particular, our reduction reduces $\GRPHP$ reduces to the \DDHLF problem, as we describe in \Cref{cor:mainreduction}.
So far this shows that one instance of $\GRPHP$ reduces to the \DDHLF problem. We then show, in \Cref{lem:ddhlf_sum}, that we can embed multiple instances of \DDHLF in parallel into one instance of \DDHLF. Hence $\PGRPHP$ reduces to \DDHLF as well, and hence \Cref{thm:PGRPHP} implies \Cref{thm:main}.

\subsection{Additional results}

We also consider the question of showing a separation between $\QNC^0$ and $\AC^0[2]$, where $\AC^0[2]$ is $\AC^0$ with unbounded fan-in XOR gates. We implement the first two steps of the strategy above, where we come up with a problem in $\QNCcat$ that cannot be solved by an $\AC^0[2]$ circuit, even on a $o(1)$ fraction of the inputs. 
But we do not know how to remove the reliance on the cat state in this setting.

\para{Parity Bending Problem.} 
In the Parity Bending Problem, which we denote $\PBP_n$, we are given a string $x\in\B^n$, and our goal is to output a string $y\in\B^n$ such that if $|x| \equiv 0 \pmod 3$ then $|y|\equiv 0 \pmod 2$, and if $|x|\in\{1,2\} \pmod 3$ then $|y|\equiv1 \pmod 2$. 
Our main result is \Cref{thm:PBP}, which says that $\PBP$ can be solved with high probability by a $\QNCcat$ circuit, but needs exponential-size $\AC^0[2]$ circuits to solve with probability significantly greater than half.

\begin{thm}[PBP]\label{thm:PBP}
There exists a $\QNCcat$ circuit that solves the Parity Bending Problem $(\PBP_n)$ on any input with probability $\geq\frac{3}{4}$. 
But there exists an input distribution on which any $\AC^0[2]/\rpoly$ of depth $d$ and size at most $\mathrm{exp}\bigl(n^{\frac{1}{10d}}\bigr)$ only solves the problem with probability $\frac{1}{2}+\frac{1}{n^{\Omega(1)}}$. 
\end{thm}

As with the Parity Halving Problem, the quantum circuit that solves this problem with bounded error is very simple as shown in \Cref{sec:PBP}. 
For this problem, the classical lower bound is easier to show than before, and follows from the work of Razborov and Smolensky~\cite{Raz87,Smo87}, which shows that $\AC^0[2]$ circuits correlate poorly with the Mod 3 function. 

As before, we can strengthen the separation to make the quantum circuit's success probability arbitrarily close to $1$ and the classical circuit's success probability arbitrarily close to $0$ by defining a new version of the Parity Bending Problem that we call the Parallel Parity Bending Problem. In this problem, we are given many instances of the Parity Bending Problem, and required to solve at least $2/3$ of them. Since $\QNCcat$ can solve this problem with probability $3/4$, it can solve more than $2/3$ of the instances with high probability.

\begin{thm}[Parallel PBP]\label{thm:RPBP}
The Parallel Parity Bending Problem can be solved with probability $1-o(1)$ by a $\QNC/\qpoly$ circuit, but any $\AC^0[2]/\rpoly$ circuit can only solve the problem with probability $\frac{1}{n^{\Omega(1)}}$. 
\end{thm}

At a high level this lower bound proceeds similar to \Cref{thm:PGRPHP}, again employing Vazirani's XOR lemma~\cite{Vaz86}, but there are technical difficulties caused by the fact that the Boolean version of the Mod 3 function is unbalanced and easy to compute on a $2/3$ fraction of the inputs.

\subsection{Discussion and open problems}
Our main results show that there is a search problem (either the \DDHLF problem or the $\PGRPHP$) in $\QNC^0$ that is not in $\AC^0$, and that there is a search problem ($\PPBP$) in $\QNCcat$ that is not in $\AC^0[2]$. 
One open problem is to generalize both separations and show that there is a search problem in $\QNC^0$ that is not in $\AC^0[2]$, or more generally $\AC^0[p]$ for any prime $p$.
This is essentially the frontier of circuit lower bounds, and it will be difficult to go further without radically new techniques. 

One could try to achieve a quantum advantage using even weaker classes than $\QNC^0$ or classes incomparable to $\QNC^0$. 
The recent result of Raz and Tal~\cite{RT18} exhibits a decision problem in $\mathsf{BQLOGTIME}$ (bounded-error quantum logarithmic time) that is not in $\AC^0$. 
Note that as classes of search problems, $\mathsf{BQLOGTIME}$ and $\QNC^0$ are incomparable, since both can solve search problems the other cannot.

\section{Parity Halving Problem} \label{sec:PHP}

Recall the Parity Halving Problem from the introduction. We now define a more general version of the problem with $n$ input bits and $m$ output bits.

\begin{problem}[Parity Halving Problem, $\PHP_{n,m}$] \label{prob:PHP}
Given an input $x \in \{ 0, 1 \}^{n}$ of even parity, output a string $y \in \{ 0, 1 \}^{m}$ such that 
\begin{equation}
|y| \equiv \frac{1}{2}|x| \pmod{2}.    
\end{equation}
Alternately, $y$ must have even parity if $|x| \equiv 0 \pmod{4}$ and odd parity if $|x| \equiv 2 \pmod{4}$. We also define $\PHP_n$ to be $\PHP_{n,n}$.
\end{problem}

The main result of this section is to show this problem is in $\QNC^{0}/\Cat$, but not $\AC^{0}/\rpoly$. We now restate this result (\Cref{thm:PHP}) more formally:

\begin{repthm}{thm:PHP}
The Parity Halving Problem $(\PHP_n)$ can be solved exactly by a depth-$2$, linear-size quantum circuit starting with the $\ket{\Cat_n}$ state. 
But on the uniform distribution over all valid inputs (even parity strings), any $\AC^0/\rpoly$ circuit of depth $d$ and size $s \leq \mathrm{exp}\bigl(n^{\frac{1}{2d}}\bigr)$ only solves the problem with probability $\frac{1}{2} + \exp(-n^{1-o(1)} \big/ O(\log s)^{2(d-1)})$. 
\end{repthm}

We prove this theorem in several parts. 
First we prove the quantum upper bound in \Cref{sec:PHPquantum} (\Cref{thm:PHPquantum}).
The lower bound on $\AC^0$ circuits via a sequence of incrementally stronger lower bounds, culminating in the claimed lower bound. 
We start in \Cref{sec:NClocalityone} by showing a lower bound (\Cref{thm:parityhalvinggame}) for a very simple class of circuits, $\NC^0$ circuits of locality $1$, i.e., $\NC^0$ circuits where every output is an arbitrary function of exactly one input bit.
We then extend the lower bound to arbitrary $\NC^0$ circuits in \Cref{sec:NC} (\Cref{thm:NC0}), and to $\AC^0$ circuits in \Cref{sec:AC} culminating in the \ACz lower bound for $\PHP_{n,m}$ in \Cref{thm:AC0}, from which the lower bound in \Cref{thm:PHP} follows straightforwardly by setting $m=n$. 

\subsection{Quantum upper bound}
\label{sec:PHPquantum}

Before we get into the details of the proof, let us motivate the problem.
Observe that the problem naturally defines an interesting $n$-player cooperative non-local game, which we call the Parity Halving Game.
In this game, there are $n$ players, and each player gets one of the $n$ input bits and outputs a single bit, with no communication with the other players.
The input and output conditions are the same as in $\PHP_n$: The input is promised to be of even Hamming weight, and the players win the game if their output's parity satisfies the condition in \Cref{prob:PHP}.

Because the players are not allowed to communicate, the strategies permitted in the non-local game are far more restricted than an $\AC^0$ circuit or even an $\NC^0$ circuit for $\PHP_n$ since each output bit is only allowed to depend on one input bit. 
We will call this model $\NC^0$ with locality $1$.

Now that we have defined a game, we can study the probability of success for classical players versus the probability of success for quantum players who share entanglement before the game begins. 
In fact, when $n = 3$, the Parity Halving Game coincides with the well-known Greenberger--Horne--Zeilinger (GHZ) game~\cite{GHZ89}. 
It is known that quantum players sharing entanglement, and specifically the state $|\Cat_3\>=\frac{1}{\sqrt{2}}(\ket{000} + \ket{111})$, can always win the GHZ game with certainty, but classical players can win the GHZ game with probability at most $3/4$.

This $n$-player generalization of the GHZ game is very natural and quantum players can win the Parity Halving Game exactly using a $|\Cat_n\>$ state. This game has been studied before, and we are aware of two other works that analyze this game: the first by Mermin \cite{Mermin90}, and the second by Brassard, Broadbent, and Tapp \cite{BBT05}. Both papers exhibit the quantum strategy that wins perfectly and argue that classical strategies fail (as we do in the next section). 

The strategy for winning the $3$-player GHZ game generalizes to yield a perfect strategy for winning the $n$-player game as well, which yields a depth-$2$ linear-size quantum circuit for $\PHP_n$. We now describe the quantum strategy and the corresponding constant-depth quantum circuit.

\begin{thm}[Quantum circuit for $\PHP_n$]
\label{thm:PHPquantum}
The Parity Halving Problem $(\PHP_n)$ can be solved exactly by a depth-$2$, linear-size quantum circuit starting with the $\ket{\Cat_n}$ state.
\end{thm}
\begin{proof}
We describe this circuit in the language of the $n$-player Parity Halving Game described above.
The circuit is depicted in \Cref{fig:php} (on page \pageref*{fig:php}).  
Let the input to the $i^\mathrm{th}$ player in the Parity Halving Game be called $x_i$, and their output be called $y_i$.
In our protocol, the players will share an $n$-qubit cat state $\ket{\Cat_n} = \frac{1}{\sqrt{2}}\left(\ket{0^n} + \ket{1^n} \right)$, and each player receives one qubit of the cat state at the beginning.

Each player starts by applying a phase gate, $S = \bigl(\begin{smallmatrix} 1 & 0 \\ 0 & i \end{smallmatrix}\bigr)$, to their qubit of the cat state if their input bit is $1$. If their input bit is $0$, they do nothing.
In other words, the player applies a control-$S$ gate with $x_i$ as the source and their qubit of the cat state as the target.
After this step, the cat state has been transformed to 
\begin{equation}
\frac{1}{\sqrt{2}}\left(\ket{0^n} + i^{|x|} \ket{1^n} \right).    
\end{equation}
But since $x$ has even parity, this state is either $\ket{\Cat_n}$ or the ``minus cat state'' $\frac{1}{\sqrt{2}}\left(\ket{0^n} - \ket{1^n} \right)$. We will denote this state by $Z\ket{\Cat_n}$ since this is the state one obtains by applying the $Z=\bigl(\begin{smallmatrix} 1 & 0 \\ 0 & -1 \end{smallmatrix}\bigr)$ gate to any one qubit of the cat state.
When $|x|\equiv 0 \pmod 4$, this state will be $|\Cat_n\>$ and when $|x|\equiv 2 \pmod 4$, this will be $Z\ket{\Cat_n}$.
Note that $\ket{\Cat_n}$ and $Z\ket{\Cat_n}$ are orthogonal states.

Finally, each player applies the Hadamard gate $H=\frac{1}{\sqrt{2}}\bigl(\begin{smallmatrix} 1 & 1 \\ 1 & -1 \end{smallmatrix}\bigr)$ to their qubit of the cat state, measures the qubit, and outputs that as $y_i$. The operator $H^{\otimes n}$ maps the cat state $|\Cat_n\>$ to a uniform superposition over even parity strings, and maps $Z|\Cat_n\>$ to a uniform superposition over odd parity strings. This follows from the following equations:
\begin{align}
    H^{\otimes n} |0^n\> = \frac{1}{\sqrt{2^n}}\sum_{x\in\B^n} |x\>, \qquad \mathrm{and} \qquad 
    H^{\otimes n} |1^n\> = \frac{1}{\sqrt{2^{n}}}\sum_{x\in\B^n} (-1)^{|x|}|x\>.
\end{align}

Thus, when the players measure their qubits, they will get either a random even parity string when $|x|\equiv 0 \pmod 4$ or a random odd parity string when $|x|\equiv 2 \pmod 4$, as desired.
\end{proof}  

Note that the idea of inducing a relative phase proportional to the Hamming weight of a string is studied more generally and called ``rotation by Hamming weight'' in \cite{HS05}.

\subsection{Lower bound for \texorpdfstring{$\NC^{0}$}{NC0} circuits of locality 1}
\label{sec:NClocalityone}

We now discuss the success probability of classical strategies for the Parity Halving Game. 
This was already studied by Mermin \cite{Mermin90}, and Brassard, Broadbent, and Tapp \cite{BBT05}. 
Both papers argue that classical strategies only succeed with probability exponentially close to $1/2$ on the uniform distribution over even-parity inputs.

We reprove these lower bounds on the Parity Halving Game and also prove lower bounds for a restricted version of the game. In the restricted version of the game we only consider inputs consistent with some restriction of the input bits, i.e., where the values of some input bits have been fixed and are known to all the players,
and we only consider all even-parity inputs consistent with this fixing of input bits. We need this generalization later on in the proof since some input bits will be fixed by a random restriction in the  $\AC^0$ lower bound argument.

\begin{thm}[Classical lower bound for Parity Halving Game]
\label{thm:parityhalvinggame}
On the uniform distribution over even-parity strings, the success probability of any classical strategy for the Parity Halving Game with $n$ players is at most $\frac{1}{2} + 2^{-\ceil{n/2}}$.

Now consider the \emph{restricted} Parity Halving Game with $n$ players, where $d$ of the input bits have fixed values known to all players. On the uniform distribution over even-parity strings consistent with the fixed input bits, the success probability of any classical strategy is at most  $\frac{1}{2} + 2^{- \ceil{(n-d)/2}}$.
\end{thm}

\begin{proof}
We start with the lower bound for the unrestricted Parity Halving Game.
Since we consider classical strategies against a fixed input distribution, we can without loss of generality only consider deterministic strategies.
This is because a randomized strategy is simply a probability distribution over deterministic strategies, and we can pick the strategy that does the best against the chosen input distribution. (This is the easy direction of Yao's minimax principle.)

Since each player only has one input bit $x_i$, and one output bit $y_i$, there are only four deterministic strategies: output $y_i = 0$, $y_i = 1$, $y_i = x_i$, or $y_i = x_i \oplus 1$. In any case, each $y_i$ is a degree-$1$ polynomial (over $\mathbb F_2$) in $x_i$. It follows that the parity of the outputs, $\bigoplus_{i=1}^{n} y_i$, can be expressed as multivariate linear polynomial in $x_1, \ldots, x_n$, say $a + b \cdot x$ for some $a \in \mathbb F_2$ and $b \in \mathbb F_2^n$. 
We want to upper bound the success probability of any such strategy.

Now consider the function $f(x) = \Re(i^{|x|})$. We have 
\begin{equation}
f(x) = \begin{cases}
1 & \text{if $|x|\equiv 0 \pmod{4}$} \\
-1 & \text{if $|x|\equiv 2 \pmod{4}$} \\
0 & \text{otherwise.}
\end{cases}    
\end{equation}
The function $f(x)$ matches the parity of the output bits (as $\pm 1$) of the $\PHP_n$ function on an input $x$. 
More precisely, $f(x)$ gives the correct parity (as $\pm 1$) when $x$ satisfies the promise of $\PHP_n$, and evaluates to $0$ for inputs outside the promise.

It follows that the product $(-1)^{a + b \cdot x} f(x)$ is $1$ if the strategy corresponding to $a + b \cdot x$ is correct, $-1$ if it is incorrect, and $0$ on inputs that are outside the promise. We define the correlation $\chi$ of a classical strategy as the absolute value of the fraction of valid inputs on which it is correct minus the fraction of valid inputs on which it is incorrect. We can compute this quantity as follows:
\begin{align}
\chi 
&= \left| \mathop{\mathbb{E}}_{\substack{x\in \mathbb F_2^n: \sum_i x_i = 0}} \left[(-1)^{a + b \cdot x} f(x)\right] \right| \\
&= \left| \frac{1}{2^{n-1}} \sum_{x \in \mathbb F_2^n} (-1)^{a + b \cdot x} \Re(i^{|x|}) \right| \\
&\leq \frac{1}{2^{n-1}} \left| \Re\left( \sum_{x \in \mathbb F_2^n} (-1)^{b_1 x_1 + \cdots + b_n x_n} \cdot i^{x_1 + \cdots + x_n} \right) \right| \\
&= \frac{1}{2^{n-1}} \left| \Re\left( \sum_{x_1 \in \mathbb F_2} (-1)^{b_1 x_1} i^{x_1} \cdots \sum_{x_n \in \mathbb F_2} (-1)^{b_n x_n} i^{x_n}\right) \right|  \\
&= \frac{1}{2^{n-1}} \left|\Re \left((1 + i^{1 + 2b_1}) \cdots (1 + i^{1+2b_n})\right)\right|.
\end{align}
That is, we want to know the real part of a product of $n$ terms, each of which is $1 \pm i$. Since $1 \pm i$ is $\sqrt{2}$ times a primitive eighth root of unity, the product is $2^{n/2}$ times an eighth root of unity. After factoring out the $\sqrt{2}$ from each term, we have to determine the possible values of the product of $n$ numbers of the form $\frac{1}{\sqrt{2}}(1\pm i)$. When $n$ is even, their product must lie in the set $\{\pm 1, \pm i\}$, and when $n$ is odd it must lie in the set $\frac{\pm 1 \pm i}{\sqrt{2}}$. In both cases, we see that the real part of the product is either $0$ or $\pm 2^{\lfloor n/2 \rfloor}$, so the correlation is $\chi = 0$ or $\chi = 2^{-\lceil n/2 \rceil+1}$. Since the success probability is $(1+\chi)/2$, this proves the first part of the theorem. 

Now let us move on the to restricted version of the game and fix some of the inputs. If some individual bit $x_j$ is restricted, then the term $\sum_{x_j \in \mathbb F_2} (-1)^{b_j x_j} i^{x_j}$ in the analysis above becomes either $1$ or $(-1)^{b_j} i$. This term is a fourth root of unity, so it does not contribute to the magnitude of the product, since the fourth roots of unity have magnitude 1. Furthermore, it does not change the set of potential phases, since both the sets above are invariant under multiplication by a fourth root of unity. Since the constraint also halves the number of possible inputs, the effect on the correlation is the same as just removing that bit. In other words, $\chi$ is at most $2^{- \lceil (n-d)/2  \rceil + 1}$. It follows that the success probability of a classical strategy is
\begin{equation}
\frac{1+\chi}{2} = \frac{1}{2} + 2^{-\lceil (n-d)/2 \rceil}.     \qedhere
\end{equation}
\end{proof}
It is interesting to note that for the unrestricted game, Brassard, Broadbent, and Tapp \cite{BBT05} show that there are strategies matching this upper bound.

\subsection{From \texorpdfstring{$\NC^{0}$}{NC0} circuits of locality 1 to general \texorpdfstring{$\NC^{0}$}{NC0} circuits}
\label{sec:NC}

We can view $\NC^{0}$ circuits as a more powerful model of computation than the game considered in the previous section. Now each player is allowed to look at the input bits of a constant number of other players before deciding what to output. 
For example, the players could band together into constant-sized groups and look at all the other bits in the group to make a slightly more informed choice. 
However, intuitively it seems that the players cannot do much better than before. We will show this formally by proving that $\NC^{0}$ circuits cannot solve $\PHP_n$.

First, we define some terms. Fix a circuit $C$ and define the \emph{interaction graph} of the circuit $C$ to be a bipartite graph on the input bits and output bits where there is an edge from an input bit $x_i$ to an output bit $y_j$ if there is a path from $x_i$  to $y_j$ in the circuit $C$ (i.e., if $x_i$ can affect $y_j$ in $C$). 
The neighborhood of a vertex in this graph is sometimes called its \emph{light cone}. That is, the light cone of an output bit, $\outputlc{y_i}$, is the set of input bits which can affect it, and the light cone of an input bit, $\inputlc{x_i}$ is the set of output bits which it can affect. 
For example, if all gates have fan-in $2$, then the light cone of any output bit in a circuit of depth $d$ is of size at most $2^d$. In general, we say that a circuit $C$ has locality $\ell$ if the light cone of any output bit is of size at most $\ell$.

Note that while the fan-in of gates sets an upper bound on the light cone of an output bit, the fan-out sets an upper bound for the light cone of input bits. In all the classical circuit classes we study in this paper, fan-out is unbounded, hence even in a constant-depth circuit one input bit can affect all output bits.

\begin{prop}
\label{prop:independentinputs}
Let $C$ be a circuit with $n$ inputs, $m$ outputs, and locality $\ell$. 
There exists a subset of inputs bits $S$ of size $\Omega\left(\min\left\{n, \frac{n^2}{\ell^2 m}\right\}\right)$ such that each output bit depends on at most one bit from $S$.
\end{prop}
\begin{proof}
Since each output bit has a light cone of size at most $\ell$, the interaction graph has at most $\ell m$ edges. 
This implies that, on average, an input bit has a light cone of size $\frac{\ell m}{n}$. 
Our goal is to find a set of input bits $S$ such that their light cones are pairwise disjoint, since then the light cone of any output contains at most one element of $S$. 

Consider the intersection graph between input variables. That is, we consider the graph on $x_1, \ldots, x_n$, where $x_i$ is connected to $x_j$ if their light cones intersect. A variable $x_i$ that had degree $d$ in the original graph has degree at most $d\ell$ in the intersection graph, since each output vertex has locality $\ell$.
Hence the average degree in the intersection graph, denoted by $D$, is at most $\frac{\ell^2 m}{n}$.
By Tur\'{a}n's theorem, in any graph on $n$ vertices with average degree at most $D$ there exists an independent set of size at least $n/(1+D)$. 
Thus, we get a set $S\subseteq \{x_1, \ldots, x_n\}$ of size $\Omega\left(\min\left\{n, \frac{n^2}{\ell^2 m}\right\}\right)$ such that the light cones of every pair of input bits in $S$ do not intersect.
\end{proof}

We are now ready to prove a lower bound on \NCz circuits of locality $\ell$ solving $\PHP_{n,m}$.

\begin{thm}[PHP is not in $\NC^0$]\label{thm:NC0}
Let $C$ be an \NCz circuit with $n$ inputs, $m$ outputs, and locality $\ell$. 
Then $C$ solves $\PHP_{n,m}$ on a random even-parity input with probability at most $\frac{1}{2} + 2^{-\Omega\left(\min\left\{n, \frac{n^2}{\ell^2 m}\right\}\right)}$.
\end{thm}
\begin{proof}
Let the circuit $C$ solve $\PHP_{n,m}$ on a random even-parity input with probability $p$.
By the previous theorem, there is a set of input bits with disjoint light cones, $S$, and $|S| = \Omega\left(\min\left\{n, \frac{n^2}{\ell^2 m}\right\}\right)$. For the remainder of this proof fix any such $S$. 

Now consider choosing an arbitrary assignment for the bits outside $S$ and running the circuit $C$ on the distribution of random even-parity strings consistent with this arbitrary assignment. The probability of success of circuit $C$ may depend on the arbitrary assignment chosen, but since the success probability for a random choice is $p$, there exists one assignment for which the success probability is at least $p$. Let us fix this assignment of bits outside $S$. Now we have an assignment for bits outside $S$ such that $C$ is correct with probability at least $p$ on a random even-parity input consistent with this assignment. 

We will now argue that the circuit gives a strategy for the restricted Parity Halving Game on $n$ players with $n-|S|$ restricted bits with probability of success at least $p$. To do so, we assign a player for every input bit. 
Only the players assigned to bits in $S$ will have unrestricted inputs. 
Since the light cones of bits in $S$ do not intersect, a player with input bit in $S$ can compute the values of all outputs in its light cone (since all the bits outside $S$ are fixed and known to everyone). 
This player can now output the parity of all these output bits.
Some output bits may not appear in any input light cone; we add the parity of these bits to an arbitrary player's output. 
Now the the parity of the players' outputs is the same as the parity of the circuit's output. 
This gives a classical strategy for the restricted Parity Halving Game with $n$ players and $n-|S|$ restricted bits with success probability at least $p$. Finally, from \Cref{thm:parityhalvinggame} we get that $p \leq \frac{1}{2} + 2^{-\Omega\left(\min\left\{n, \frac{n^2}{\ell^2 m}\right\}\right)}$.
\end{proof}

Note that this theorem is essentially tight. It says that to achieve a high probability of success, we need $n^2 = \Theta(\ell^2 m)$. We can indeed achieve success probability 1 at both extremes: when $m = \Theta(n^2)$ and $\ell=2$, or when $m=1$ and $\ell=n$. For the first setting of parameters, as noted in the introduction, there is a simple depth-$1$ \NCz circuit of locality 2 that solves the problem when $m = \binom{n}{2}$. The second parameter regime is even simpler, since any Boolean function can be computed by an \NCz circuit of locality $\ell=n$. 

\subsection{From \texorpdfstring{$\NC^{0}$}{NC0} circuits to \texorpdfstring{$\AC^{0}$}{AC0} circuits}
\label{sec:AC}

In this section we finally extend our lower bound to $\AC^0$ circuits as stated in \Cref{thm:PHP}.

To do this, we use a technical tool known as a switching lemma \cite{FSS84, Ajtai83, Yao85, Has86}. Informally, a switching lemma says that with high probability randomly restricting a large fraction of the input bits to an $\AC^0$ circuit produces a circuit with small locality.

Average-case reductions from $\NC^0$ to $\AC^0$ have previously appeared in the literature (cf. \cite{Viola14}), based on the original switching lemma~\cite{Has86}. 
In this paper, we will use multi-switching lemmas, which  handle multiple output circuits much better, and were recently proved by H{\aa}stad~\cite{Hastad14} and Rossman~\cite{Rossman2017}.
Using the multi-switching lemmas instead of H{\aa}stad's original switching lemma \cite{Has86} allows us to improve the parameters dramatically.%
\footnote{
Based on the original switching lemma, we can show that $\PHP$ is hard to compute by $\AC^0$ circuits on more than $1/2 + 1/n^{\Omega(\log n)}$ of the inputs. On the other hand, based on the multi-switching lemmas, we will show  that $\PHP$ is, in fact, hard to compute on more than  $1/2 + \mathrm{exp}\bigl(-n^{1-o(1)}\bigr)$ of the inputs.}

\subsubsection{Preliminaries}
We start with some definitions. In the following, we consider restrictions and random restrictions. A restriction $\rho \in \{0,1,*\}^n$ defines a partial assignment to the inputs of a Boolean string of length $n$.
For $i=1, \ldots, n$, when $\rho_i \in \{0,1\}$ we say that the restriction fixes the value of the $i$-th coordinate, and when $\rho_i = *$ we say that the restriction keeps the $i$-th coordinate alive.

 A $p$-random restriction is a restriction sampled according to the following process: for each $i=1, \ldots, n$ independently, sample $\rho_{i}=*$ with probability $p$, $\rho_{i}=0$ with probability $(1-p)/2$ and $\rho_{i}=1$ with probability $(1-p)/2$. We denote by $\mathbf{R}_p$ the distribution of $p$-random restrictions.
 
For a Boolean function $f: \{0,1\}^n \to \{0,1\}^m$ we denote by $f|_{\rho}: \B^n \to \B^m$ the restricted function defined by
\begin{equation}
f|_{\rho}(x) = f(y) \qquad \text{where} \qquad y_i = \begin{cases}
x_i & \rho_i = * \textrm{ and}\\
\rho_i & \textrm{otherwise.}
\end{cases}
\end{equation}

\newcommand{\Rest}{\mathbf{R}}
\newcommand{\Rp}{\mathbf{R}_p}
\newcommand{\DT}{\mathrm{DT}}
\newcommand{\CKT}{\mathrm{CKT}}
\newcommand{\N}{\mathbb{N}}

Next, we give the standard definition of a decision tree. For an excellent survey on this topic, please see \cite{BuhrmanW02}.

\begin{defn}[Decision Tree]
A {\sf decision tree} is a rooted ordered binary tree $T$, where each internal node of $T$ is labeled with a variable $x_i$ and each leaf is labeled with a value 0 or 1. Given an input $x\in \B^n$, the tree is evaluated as follows. Start at the root. If this is a leaf then stop. Otherwise, query the variable $x_i$ that labels the root. If $x_i=0$, then recursively evaluate the left subtree, if $x_i=1$ then recursively evaluate the right subtree. The output of the tree is the value (0 or 1) of the leaf that is reached eventually. Note that an input $x$ deterministically determines the leaf reached at the end, and thus the output.
We say a decision tree {\sf computes} $f$ if its output equals $f(x)$, for all $x\in\B^n$. 
The complexity of such a tree is its depth, i.e., the number of queries made on the worst-case input.
We denote by $\DT(t)$ the class of functions computed by decision trees of depth at most $t$.
\end{defn}

Note that the decision tree complexity of a function $f$ is also called the deterministic query complexity of $f$.

\begin{defn}[$\mathcal F$-Decision Tree]
Suppose $\mathcal F$ is a class of functions  mapping $\{0,1\}^n$ to $\{0,1\}^m$. An $\mathcal F$-partial decision tree is a standard decision tree, except that the leaves are marked with functions in $\mathcal F$ (instead of constants).
Given an input $x\in \B^n$, the $\mathcal F$-Decision Tree is evaluated as follows. Starting from the tree's root, we go along the path defined by the input $x$ until we reach a leaf.  Then, we evaluate the function $f_v \in \mathcal F$ that labels the leaf $v$ on the input $x$, and output its value, $f_v(x)$. We denote by $\DT(t) \circ \mathcal F$ the class of functions computed by $\mathcal F$-decision trees of depth at most $t$.
 \end{defn}
 
 Note that $\mathcal F$-decision trees compute functions from $\{0,1\}^n \to \{0,1\}^m$ where $n$ and $m$ are the input and output lengths for the functions in $\mathcal F$, respectively.
 
 \begin{defn}[Tuples of functions classes]
Suppose $\mathcal F$ is a class of functions mapping $\{0,1\}^n$ to $\{0,1\}$.
We denote by ${\mathcal F}^m$ the class of functions  $F:\B^n \to \B^m$ of the form $F(x) = (f_1(x), f_2(x), \ldots, f_m(x))$, where each $f_i\in \mathcal F$. That is, ${\mathcal F}^m$ is  the class of $m$-tuples of functions in $\mathcal F$.
\end{defn}

\subsubsection{The multi-switching lemma}
The main lemma that we are going to use is a slight adaption of Rossman's lemma \cite{Rossman2017}, which combines both switching lemmas of H{\aa}stad \cite{Has86,Hastad14}.
The lemma claims that a multi-output $\AC^0$ circuit mapping $\B^n \to \B^m$ would reduce under a random restriction, with high probability, to a function in the class $\DT(2t) \circ \DT(q)^m$ (for some parameters $t$ and $q$).

Let us pause for a second to spell out what is the class $\DT(2t) \circ \DT(q)^m$.
This is the class of depth-$2t$ decision trees, whose leaves are labeled by $m$-tuples of depth-$q$ decision trees, one per output bit. 
In other words, these are functions mapping $\B^n$ to $\B^m$ that can be evaluated by adaptively querying at most $2t$ coordinates globally, after which each of the $m$ output bits can be evaluated by making at most $q$ additional adaptive queries. Note that while the first $2t$ queries are global, the last $q$ queries could differ from one output bit to another. We would typically set the parameters so that $t$ is much larger than $q$ (for example, $t = n^{1-o(1)}$ and $q = o(\log n)$).

\begin{lem}[Multi-switching lemma]\label{cor:Rossman}
Let $f:\B^n \to \B^m$ be an $\AC^0$ circuit of size $s$, depth $d$.
Let $q \in \N$ be a parameter, and set $p = 1 / (m^{1/q} \cdot O(\log s)^{d-1})$.
Then 
\begin{equation}
\forall{t}: \Pr_{\rho \sim \Rp}[f|_{\rho} \notin \DT(2t) \circ \DT(q)^m] 
\le 
s \cdot 2^{-t}.
\end{equation}
\end{lem}

We defer the proof of \Cref{cor:Rossman} to \Cref{app:Rossman} as this is an adaptation of Rossman's lemma \cite{Rossman2017}.

We would use the lemma as follows. 
First, we apply a $p$-random restriction that reduces the $\AC^0$ circuit to a $\DT(2t) \circ \DT(q)^m$ function with high probability.
Then, we further query at most $2t$ coordinates, and fix their values, by following a path in the common partial decision tree.
After which, the restricted function would be an $m$-tuple of depth-$q$ decision trees. Then, using the simple fact that a depth-$q$ decision tree is a function with locality at most $2^q$, we reduced an $\AC^0$ circuit to an $\NC^0$ circuit with locality at most $2^{q}$ with high probability.

\para{On the choice of parameters.}
We have the freedom to choose $q$ and $t$ when applying \Cref{cor:Rossman} in \Cref{thm:AC0}. 
First, we discuss the choice of $q$.
We would like the lemma to yield on one hand an $\NC^0$ circuit with small locality, and on the other hand to keep many input variables alive.
To get small locality, $q$ should be small, say $q = o(\log n)$. To keep many variables alive, $pn$ should be large, and since $p = 1/O(m^{1/q} (\log s)^{d-1})$, we would like $q$ to be large, say $q = \omega(1)$. Balancing these two requirements leads to the choice $q = \Theta(\sqrt{\log n})$.

Once $q$ is set, we would like to make $t$ as large as possible, as it controls the failure probability in \Cref{cor:Rossman}, but on the same time we want the number of alive variables after the two-step restriction process above to remain high. Since this number is roughly $pn - t$ we would choose $t$ to be a small constant fraction of $pn$ (which is $n^{1-o(1)}$). With these choices, we would be left with at least $\Omega(pn)$ variables alive and locality at most $2^{q}$ with extremely high probability.

\subsubsection{\texorpdfstring{$\AC^{0}$}{AC0} lower bound}
\begin{thm}[PHP is not in $\AC^0$]
\label{thm:AC0}
Let $n \le m \le n^2$.
Any $\AC^{0}/\rpoly$ circuit $F$ of depth $d$ and size $s \leq \mathrm{exp}\bigl(n^{\frac{1}{2d}}\bigr)$ solves $\PHP_{n,m}$ on the uniform distribution over valid inputs (even parity strings) with probability at most $\frac{1}{2} + \exp(-n^2 \big/ \bigl(m^{1+o(1)} \cdot O(\log s)^{2(d-1)}\bigr))$.
\end{thm}
\begin{proof}
Since we have a fixed input distribution, we can without loss of generality prove the lower bound against an $\AC^0$ circuit (instead of an $\AC^0/\rpoly$ circuit), since an $\AC^0/\rpoly$ circuit defines a distribution over $\AC^0$ circuits and we can simply pick the one that does the best against our input distribution.

Suppose that $F$ solves the Parity Halving Problem on a random even-parity input with probability $\frac{1}{2} + \eps$. We shall show that $\eps \le \exp(-n^2/(m^{1+o(1)} \cdot O(\log s)^{2(d-1)}))$.

We defer the choice of $q$ for later, to optimize the parameters. 
However, we will use the fact that $q = o(\log m)$ and that $q=\omega(1)$.
We set
\begin{equation} \label{eq:param-setting}
p = 1/(m^{1/q} \cdot O(\log s)^{d-1}), \qquad t = pn/8.
\end{equation}
Note that under this choice of parameters $s \le 2^{t/2}$, by the following calculation:
using the assumption that $s < \exp(n^{1/2d})$ twice, we have 
\begin{equation} \label{eq:s_vs_t}2^{t/2} = \exp(\Omega(pn)) = \exp(\Omega(m^{-o(1)} \cdot n^{(d+1)/2d}))  \gg \exp(n^{1/2d}) > s.\end{equation}

Let $\rho$ be a $p$-random restriction. 
Denote by $\mathcal{E}$ the event that: 
\begin{enumerate}
\item $F|_{\rho} \in \DT(2t) \circ \DT(q)^{m}$.
\item $\rho$ keeps alive at least $pn/2$ variables.
\end{enumerate}
Using \Cref{cor:Rossman} and Eq.~\eqref{eq:s_vs_t}, 
Item~1 holds with probability at least $1-s \cdot 2^{-t}  \ge 1-2^{-t/2} \ge 1-\exp(-\Omega(pn))$.
Item~2 holds with probability at least $1-\exp(-\Omega(pn))$ by Chernoff's bound.
Thus, by a simple union bound \begin{equation}\label{eq:good rest prob}\Pr[\mathcal{E}] \ge 1-\exp(-\Omega(pn)).\end{equation}

If $\Pr[\mathcal{E}] \le 1-\eps/2$, then we are done as Eq.~\eqref{eq:good rest prob} implies $\eps/2 \le \exp(-\Omega(pn))$.
Going forward, we may assume that $\Pr[\mathcal{E}] > 1-\eps/2$. 
In such a case, we claim that there exists a fixed restriction $\rho$ that satisfies $\mathcal{E}$, under which $F|_{\rho}$ solves the Parity Halving Problem on at least $1/2 + \eps/2$ fraction of the even-parity inputs consistent with $\rho$.
Assume by contradiction otherwise. Under our assumption:
\begin{itemize}
\item For restrictions satisfying $\mathcal{E}$, the  success probability of $F$ on the even-parity inputs consistent with $\rho$ is at most $1/2+\eps/2$.
\item
For other restrictions, the success probability of $F$  on the even-parity inputs consistent  with $\rho$ is at most $1$.
\end{itemize}
The key idea is that sampling a $p$-random restriction, and then sampling an input with even parity consistent with this restriction (if such an input exists), gives the uniform distribution over even-parity inputs.
Thus, under the above assumption, the probability that $F$ solves the Parity Halving Problem on a uniform input with even parity, is at most
\begin{equation}
\Pr[\mathcal{E}] \cdot (1/2+\eps/2)   + \Pr[\neg \mathcal{E}]\cdot 1  
\;<\;
1\cdot (1/2+\eps/2) + (\eps/2) \cdot 1 
\;= \;
1/2+\eps,
\end{equation}
yielding a contradiction.

We get that there exists a restriction $\rho$ keeping at least $pn/2$ of the variables alive, 
under which $F|_{\rho} \in \DT(2t) \circ \DT(q)^m$, such that the success probability of $F$ on the even-parity inputs consistent  with $\rho$ is at least $1/2 + \eps/2$.

In the next and final step, we will focus on a single leaf of the partial decision tree for $F|_{\rho}$.
For each leaf $\lambda$ consider the further restriction of  $F|_{\rho}$ according to the path leading to $\lambda$. This yields a new function, denoted $F_{\rho,\lambda} \in \DT(q)^m$. 
That is, $F_{\rho, \lambda}$ is a tuple of $m$ decision trees of depth $q$.
Moreover, for each $\lambda$, the number of variables left alive in $F_{\rho, \lambda}$ is at least $pn/2 - 2t \ge pn/4$.

We claim that there must exists a $\lambda$ such that $F_{\rho,\lambda}$ solves the Parity Halving Problem on even-parity inputs consistent with $\rho$ and $\lambda$ with probability at least $1/2+\eps/2$. This follows by an averaging argument similar to the one we performed above. Indeed, to uniformly sample an input with even-parity consistent with $\rho$, we can first uniformly sample a root-to-leaf path along the partial decision tree resulting in a leaf $\lambda$, and then uniformly sample an even-parity input consistent with $(\rho,\lambda)$. Since we succeed with probability at least $1/2 + \eps/2$ on uniform inputs with even-parity to $F_{\rho}$, we also succeed with probability at least $1/2 + \eps/2$  on the inputs to some $F_{\rho, \lambda}$.

We get that there exists a restriction defined by $(\rho,\lambda)$, leaving at least $pn/4$ variables alive, under which each output bit of $F$ can be computed as a depth-$q$ decision tree, and therefore depends on at most $2^q$ input bits.
Furthermore, $F_{\rho,\lambda}$ solves the Parity Halving Problem on even-parity inputs consistent with $\rho$ and $\lambda$ with probability at least $1/2+\eps/2$. 
Applying \Cref{thm:NC0} we get that 
\begin{equation}
\eps/2 \le \exp(-\Omega\left(\frac{(pn)^2}{m \cdot 2^{2q}}
\right)
),
\end{equation}
and by Eq.~\eqref{eq:param-setting}, plugging $p = 1/(m^{1/q} \cdot O(\log s)^{d-1})$,
\begin{align}
\eps/2 \le \exp(\frac{-n^2}{m \cdot 2^{2q}\cdot O(\log s)^{2(d-1)} \cdot m^{2/q}}).
\end{align}
Recall that we have not set $q$ yet.
To minimize $2^{2q} \cdot m^{2/q}$ we pick $q = \sqrt{\log m}$. This gives 
\begin{equation}
    \eps \le \exp(\frac{-n^2}{m \cdot O(\log s)^{2(d-1)} \cdot 2^{4\sqrt{\log m}}}),
\end{equation}
which concludes the proof.
\end{proof}

\section{Relaxed Parity Halving Problem}

In this section we deal with the issue that the $\QNC^0$ circuit for $\PHP$ (\Cref{prob:PHP}) needs a cat state, but $\QNC^0$ cannot create a cat state. 

In \Cref{sec:PoorMansCatState}, we first prove that a $\QNC^0$ circuit cannot create a cat state. 
But, as we show, $\QNC^0$ circuits can construct what we call a ``poor man's cat state."
This is a state of the form $\frac{1}{\sqrt{2}}(\ket{z} + \ket{\bar{z}})$ for some uncontrolled $z\in \B^n$ alongside classical ``side information'' about $z$ that allows us to determine it.
In \Cref{sec:RelaxedParityHalving} we show the poor man's cat state lets us solve a relaxed version of the Parity Halving Problem, which is nevertheless hard for $\AC^0$ circuits. 

The goal of this section is to establish \Cref{thm:GRPHP}, which will follow from \Cref{thm:GRPHPupper} and \Cref{thm:GRPHPlower}.

\subsection{A poor man's cat state} 
\label{sec:PoorMansCatState} 

We start by proving our claim that a $\QNC^{0}$ circuit cannot construct a cat state in constant depth.
\begin{thm}[Cat states cannot be created in $\QNC^0$]\label{thm:nocat}
Let $C$ be a depth-$d$ $\QNC^0$ circuit over the gates set of all $2$-qubit gates that maps $\ket{0^{n+m}}$ to $\ket{\Cat_n} \otimes \ket{0^m}$.
Then $d \geq \left( \log n \right) / 2$.
\end{thm}
\begin{proof}
Since $C$ is a depth-$d$ $\QNC^0$ circuit over the gate set of $2$-qubit gates, each input and output bit has a light cone of size $2^d$. Similar to \cref{prop:independentinputs}, consider the intesection graph of the first $n$ output bits, which hold the cat state. In this graph, the $n$ output bits are the vertices, and two output bits are adjacent if their light cones contain a common input. Since the maximum number of input bits in the light cone of an output bit is $2^d$, and each input bit can have at most $2^d$ output bits in its light cone, the maximum degree of an output bit in this intersection graph is $2^{2d}$. 
Assume toward a contradiction that $d < \left( \log n \right) / 2$. Then the maximum degree of the graph is less then $n$, and there must exist two disconnected vertices in the graph. 

Let two such output qubits of $\ket{\Cat_n}$ be called $y_i$ and $y_j$. These qubits depend on disjoint sets of input bits. Now we focus on the output of the the circuit $C$ on these two qubits. Since they are part of the cat state, on measuring these in the computational basis, we will see either $00$ or $11$, with equal probability. Since  the gates that are not in the light cones of these two qubits do not affect this, let us delete all these gates. Since the light cones were disjoint, we are now left with a circuit composed of two disjoint parts, one acting on a set of qubits that contains $y_i$ and another acting on a set of qubits that contains $y_j$. Consider the cut between these two sets of qubits. Observe that the initial state, $\ket{0^{n+m}}$, is separable across this cut, but the output state is correlated across this cut although we have not performed any gates that cross the cut. This is impossible, and hence $d \geq \left(\log n\right) /2 $.
\end{proof}

Now although $\QNC^0$ circuits cannot create the cat state, we show that $\QNC^0$ circuits are able to construct states of the form $\frac{1}{\sqrt{2}} \left(\ket{z} + \ket{\bar{z}}\right)$, where $z$ is some string in $\B^n$ and $\bar{z}$ the complement of $z$. Note that this state is exactly the cat state when $z=0^n$ or $z=1^n$. The circuits that create this state also output an auxiliary classical string $d$ such that $z$ can be determined from $d$, up to the symmetry between $z$ and $\bar{z}$. There is actually a family of $\QNC^0$ circuits which construct these states, which we now describe.

\begin{thm}[Poor man's cat state construction] \label{thm:PMCatState}
For any connected graph $G=(V,E)$ with maximum degree $\Delta$, there is a depth $\Delta + 2$ $\QNC^0$ circuit which outputs a $\abs{V}$ qubit state $\frac{1}{\sqrt{2}}\left(\ket{z} + \ket{\bar{z}}\right)$, along with a bit string $d\in \B^{E}$. Indexing the bits of $z$ by vertices of $V$ and the bits of $d$ by edges of $E$, $z$ and $d$ satisfy the property that 
\begin{align}
    z_u + z_v \equiv \sum_{e \in P(u,v)}d_e \pmod{2}
\end{align} 
for any two vertices $u,v\in V$, and any path $P(u,v)$ from $u$ to $v$. Note that this condition also implies that the sum of $d_e$ along any cycle in the graph is $0 \pmod 2$.
\end{thm}

\begin{proof}
 We first describe the $\QNC^0$ circuit. Begin with $\abs{V} + \abs{E}$ qubits in the state $\ket{0}$, and identify each of the qubits with either an edge or a vertex of the graph. Apply the Hadamard transform for each vertex qubit. 
 Now the state is $\ket{+}^{|V|} \otimes \ket{0}^{|E|}$.
 Then, for every edge $e = (u,v)$ in the graph, XOR the qubits indexed by $u$ and $v$ onto the edge qubit indexed by $e$ (i.e., let the edge qubit store the parity of the two vertex qubits).
 Explicitly, this can be done by implementing CNOT gates from qubits $u, v$ onto qubit $e$. (As discussed below, this can be done in $\Delta+1$ parallel local steps.) 
 Finally, measure all edge qubits in the standard basis. 
 
 To complete this proof we need to establish two claims: First, that the circuit leaves the unmeasured vertex qubits in the state $\frac{1}{\sqrt{2}}\left(\ket{z} + \ket{\bar{z}}\right)$, while the measured edge qubits give the classical bitstring $d$. Second, that the circuit can be implemented in depth $\Delta + 2$. 

We begin with the first claim. Imagine that we first only measure the $n-1$ edges of some spanning tree $T$. Before measurement, the vertex qubits were in a uniform superposition over all possible $2^n$ states. Each measurement on an edge qubit had two equally probable outcomes, and observing the result of this measurement reduced the number of states in the superposition by half. More precisely, measuring the qubit for edge $e = (u,v)$ yields a bit $d_e \in \{0,1\}$, which  gives a linear equation on the state: $z_u \oplus z_v = d_e$. Thus, after all edges in the spanning tree are measured, the vertex qubits must be left in some two state superposition. Furthermore, after the spanning tree is measured any two vertex qubits $u$ and $v$ must differ by the parity of the observed measurements on edge qubits along the path from $u$ to $v$. This shows the vertex qubits must be in the state $\frac{1}{\sqrt{2}}\left(\ket{z} + \ket{\bar{z}}\right)$, with the measurements on the edge qubits so far consistent with the requirements of \Cref{thm:PMCatState}. Now for any edge $e = (v,w)$ not in the spanning tree, the XOR measurements on the associated edge qubit is fixed to be equal to the XOR of edge qubit measurements along the path in $T$ from $v$ to $w$. This shows this measurement must also be consistent with the requirements of \Cref{thm:PMCatState} and cannot affect the state of the vertex qubits. The establishes the first claim.

The second claim is more straightforward. The first layer of our $\QNC^0$ circuit consists of Hadamard gates applied to all vertex qubits. 
It remains to show that we can implement all the desired CNOT gates in depth $\Delta + 1$. 
To show this we introduce a new graph $G'$ with $|V|+|E|$ vertices, that is obtained from $G$ by replacing each edge $e = (a,b)$ in $E$ with a vertex $v_{e}$ connected to its two end-points, $a$ and $b$.
Note the edges of $G'$ are in one to one correspondence with the CNOT gates we want to implement in our circuit. By our assumptions on $G$, $G'$ has degree at most $\Delta$, and so Vizing's theorem tells us the edges of $G'$ can be colored using at most $\Delta + 1$ colors. Since the edges in each color class are non-overlapping we can apply all the CNOT gates in one color class simultaneously, and thus apply all the CNOT gates in depth $\Delta + 1$.
\end{proof}

In the remainder of this paper, we primarily apply \Cref{thm:PMCatState} when $G$ is a spanning tree of a 2D grid, with diameter $2\sqrt{n}$ as depicted in \Cref{figure:gridCatState}, which also describes the associated $\QNC^0$ circuit. This $\QNC^0$ circuit has the nice feature that it is spatially local,\footnote{Here spatially local means here that circuit may be implemented in hardware with the qubits placed on a 2D grid and CNOT gates allowed only between neighbouring qubits.} while any bit of $z$ is specified by relatively few, $O(\sqrt{n})$, bits of $d$. This graph has constant degree $\Delta = 3$.

\begin{figure}[tbhp]
\centering
\begin{tikzpicture}
\node (1) at (0,0) {\:\:};
\node (2) at (0,-0.8) {\:\:};
\node (3) at (0,-1.6) {\:\:};
\node (4) at (0,-2.4) {\:\:};
\node (5) at (1.1,0.24) {\:\:};
\node (6) at (1.9,0.24) {\:\:};
\node (7) at (1.58, -0.8) {\:\:};
\node (8) at (1.58, -1.6) {\:\:};
\node (9) at (1.58, -2.4) {\:\:};
\node (10) at (2.7, 0.24) {\:\:};
\node[minimum height=13pt] (blankOne) at (0.33,-3.3) {};
\node[minimum height=13pt] (blankTwo) at (3.25,-0.35) {};
\node[minimum height=13pt] (blankThree) at (1.9, -3.3) {};
\node[minimum height=13pt] (blankFour) at (-.16,-3.1) {};
\node[minimum height=13pt] (blankFive) at (1.4,-3.1) {};
\node[minimum height=13pt] (blankSix) at (3.4,-0.9) {};
\node[minimum height=13pt] (blankSeven) at (3.4,-2.1) {};
\node[minimum height=13pt] (blankEight) at (3.4,-2.5) {};
\node (blankNine) at (3.6,-0.39) {};
\node[minimum height=13pt] (blankTen) at (4.0,-0.3) {};
\node[circle, draw, anchor=west, inner sep=0.07cm] (v1) at (1.east) {};
\node[circle, draw, anchor=west, inner sep=0.07cm,fill = black] (v2) at (2.east) {};
\node[circle, draw, anchor=west, inner sep=0.07cm] (v3) at (3.east) {};
\node[circle, draw, anchor=west, inner sep=0.07cm,fill = black] (v4) at (4.east) {};
\node[circle, draw, anchor=north, inner sep=0.07cm,fill = black] (v5) at (5.south) {};
\node[circle, draw, anchor=north, inner sep=0.07cm] (v6) at (6.south) {};
\node[circle, draw, anchor=west, inner sep=0.07cm,fill = black] (v7) at (7.east) {};
\node[circle, draw, anchor=west, inner sep=0.07cm] (v8) at (8.east) {};
\node[circle, draw, anchor=west, inner sep=0.07cm,fill = black] (v9) at (9.east) {};
\node[circle, draw, anchor=north, inner sep=0.07cm,fill = black] (v10) at (10.south) {};
\node[circle, inner sep=0.07cm, anchor = south] (vblankOne) at (blankOne.north) {};
\node[circle, inner sep=0.07cm, anchor = south] (vblankTwo) at (blankTwo.north) {};
\node[circle, inner sep=0.07cm, anchor = south] (vblankThree) at (blankThree.north) {};
\node[circle, inner sep=0.07cm, anchor = west] (vblankFour) at (blankFour.east) {$\vdots$};
\node[circle, inner sep=0.07cm, anchor = west] (vblankFive) at (blankFive.east) {$\vdots$};
\node[circle, inner sep=0.07cm, anchor = south] (vblankSix) at (blankSix.north) {};
\node[circle, inner sep=0.07cm, anchor = south] (vblankSeven) at (blankSeven.north) {};
\node[circle, inner sep=0.07cm, anchor = south] (vblankEight) at (blankEight.north) {};
\node[circle, inner sep=0.02cm, anchor = south] (dots) at (blankNine.north) {$\ldots$};
\node[circle, inner sep=0.07cm, anchor = south] (vblankThree) at (blankThree.north) {};
\path[draw] (v1) -- (v2) -- (v3) -- (v4) -- (vblankOne);
\path[draw] (v1) -- (v5) -- (v6) -- (v10) -- (vblankTwo);
\path[draw] (v6) -- (v7) -- (v8) -- (v9) -- (vblankThree);
\draw [decorate,decoration={brace,amplitude=10pt},yshift=0pt]
(-0.05,0.5) -- (5.5,0.5) node [yshift = 28pt, xshift=-3pt, below,black,midway] {$\sqrt{n}$};
\draw [decorate,decoration={brace,amplitude=10pt},yshift=0pt]
(-0.25,-4.4) -- (-0.25,0.2) node [black,midway, xshift=-22pt] {$\sqrt{n}$};
\end{tikzpicture}
\caption{Grid Implementation of a Poor Man's Cat State. Black vertices are ``edge" qubits, and are used to measure the parity of their neighbours. White vertices are ``vertex" qubits. They are initialized in the $\ket{+}$ state, and make up the poor man's cat state after the edge qubits are measured.}
\label{figure:gridCatState}
\end{figure}
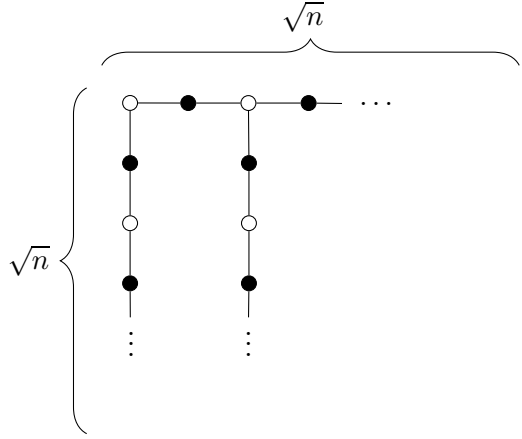

It is worth mentioning that if we relax the requirement that our implementation be spatially local, we can improve on the number of bits of $d$ required to specify any bit of $z$. In particular, applying \Cref{thm:PMCatState} to a balanced binary tree gives an output string $d$ with at most $\log n$ bits of $d$ required to specify any bit of $z$. This version of the problem would lead to slightly better parameters in \Cref{thm:GRPHP}, but then our final problem would not longer be solved by a 2D quantum circuit and would no longer reduce to the \DDHLF problem.
Hence the construction illustrated in \Cref{figure:gridCatState} will be sufficient for our purposes. 

\subsection{The Relaxed Parity Halving Problem} \label{sec:RelaxedParityHalving}

Having constructed a poor man's cat state, a natural idea would be to try and use this state instead of the cat state to solve the Parity Halving Problem. For example, we can feed the poor man's cat state into the quantum circuit (instead of a true cat state) and hope for the best. Unsurprisingly, this does not solve the Parity Halving Problem. 

However, we can define a new problem from this failed attempt, which has $\QNC^{0}$ circuits by construction. We call this problem the \emph{Relaxed Parity Halving Problem}, although we will see that the precise definition of the problem depends on how the poor man's cat state is constructed. 

\begin{thm}[PHP circuit applied to a poor man's cat state]
	\label{thm:phpbutpoorer}
	The quantum circuit for $\PHP_n$ applied to the state $\frac{\ket{z}+\ket{\overline{z}}}{\sqrt{2}}$, where $z\in\B^n$ (instead of the cat state), and an input $x \in \{ 0, 1 \}^{n}$ of even parity, yields an output string $y\in\B^n$ such that 
	\begin{equation}
	|y| \equiv \frac{1}{2}|x| + \langle z, x \rangle \pmod{2},
	\end{equation}
	where $\<z,x\> := \sum_{i\in[n]}z_i \cdot x_i$.
Note that this is the same condition as for $\PHP_n$ (\Cref{prob:PHP}) except for the addition of the $\langle z, x \rangle$ term. 
\end{thm}
\begin{proof}
Let us apply the quantum circuit solving $\PHP$ (as depicted in \Cref{fig:php}) to the poor man's cat state.
We first apply a phase gate ($S$ gate) to qubit $i$ of the poor man's cat state if $x_i = 1$. 
This yields the state
\begin{equation}
\frac{i^{\langle z, x\rangle}\ket{z} + i^{\langle \bar{z}, x\rangle} \ket{\overline{z}}}{\sqrt{2}} = i^{\langle{z,x\rangle}} \cdot \frac{\ket{z} + i^{|x| - 2\langle z, x \rangle} \ket{\overline{z}}}{\sqrt{2}} = i^{\langle{z,x\rangle}} \cdot \frac{\ket{z} + (-1)^{|x|/2 - \langle z, x \rangle} \ket{\overline{z}}}{\sqrt{2}}.
\end{equation}
Up to a global phase, which can be ignored, the state is $\frac{1}{\sqrt{2}} \bigl(\ket{z} +
(-1)^{|x|/2 - \langle z, x \rangle} \ket{\bar{z}}\bigr)$.
The next stage of the algorithm applies Hadamard gates to all input qubits. 
Thus we have
\begin{align}
H^{\otimes n} \left( \frac{\ket{z} + (-1)^{|x|/2 - \langle z, x \rangle}  \ket{\overline{z}}}{\sqrt{2}}\right) 
= \frac{1}{\sqrt{2^{n+1}}} \sum_{y\in \B^n} \left((-1)^{\langle{y,z\rangle}} + (-1)^{\langle{y,\bar{z}\rangle}} (-1)^{|x|/2 - \langle z, x \rangle} \right) \ket{y}.
\end{align}
On measuring this state in the computational basis, we get only those $y$ whose coefficient is nonzero. Hence we get a uniform distribution over all strings $y$ satisfying $\langle{y,z\rangle} \equiv \langle{y,\bar{z}\rangle} + |x|/2 - \langle{z,x\rangle} \pmod{2}$ or equivalently,
\begin{equation}
|y| \equiv \frac{1}{2}|x| + \langle z,x \rangle \pmod{2}. \qedhere
\end{equation}
\end{proof}

We now define a new problem based on this observation. 

\begin{problem}[Relaxed Parity Halving Problem for graph $G$]
\label{problem:RPHP}
Fix a connected graph $G = (V,E)$. Given an input $x \in \{ 0, 1 \}^{V}$ promised to have even parity, the \emph{Relaxed Parity Halving Problem} or RPHP outputs $y \in \{ 0, 1 \}^{V}$ and $d \in \{ 0, 1 \}^{E}$, such that there exists a $z \in \{ 0, 1 \}^{V}$ with the property
\begin{align}
&\forall (u,v) \in E, \,\, z_u \oplus z_v = d_{(u,v)}, \quad\mathrm{and}\\
&|y| \equiv \frac{1}{2}|x| + \langle z, x \rangle \pmod{2}.
\end{align}
\end{problem}

Note that in the definition above, a string $z$ satifying the first constraint exists if and only if the parity of $d$ along every cycle is $0$. 
If there are no cycles, then there always exists a $z$. 
When $z$ does exist, it is unique up to complement, which does not change the second condition in the problem statement because $\langle z, x \rangle \equiv \langle \overline{z}, x \rangle \pmod{2}$. 
To see this, recall that $x$ has even parity, and hence $0 \equiv \langle 1, x \rangle \equiv \langle z, x \rangle + \langle \overline{z}, x \rangle \pmod{2}$.

To fully specify the problem, we need a family of graphs with $|V| = n$ for infinitely many $n$. 
For this paper, we are primarily interested in the 2D grid (so that the quantum circuit is specially local) and some reasonable (say, diameter $O(\sqrt{n})$) spanning tree of this graph. 
This gives us the Grid Relaxed Parity Halving Problem below. We use a spanning tree rather than the grid graph itself because deleting edges from $G$ only makes the problem easier (since we can drop the corresponding bits of the output string $d$), which makes our lower bounds stronger. Additionally, choosing a spanning tree ensures that a string $z$ satisfying the constraints of the problem always exists. It turns out the upper bounds (i.e., $\QNC^{0}$ circuit we construct) can be easily modified to solve the problem on the entire 2D grid without deleting edges. 

\begin{problem}[Grid Relaxed Parity Halving Problem]
Consider the 2D grid of size $\sqrt{n} \times \sqrt{n}$ and fix an $n$ vertex spanning tree $G = (V,E)$ of low diameter. For concreteness, fix the spanning tree which takes the first row of edges and all columns (as depicted in \Cref{figure:gridCatState}). Then the \emph{Grid Relaxed Parity Halving Problem} is the $\RPHP$ associated with $G$. That is, given $x \in \{ 0, 1 \}^{V}$, output $y \in \{ 0, 1 \}^{V}$ and $d \in \{ 0, 1 \}^{E}$ such that 
\begin{align}
&\forall (u,v) \in E, \,\, z_u \oplus z_v = d_{(u,v)}, \quad\mathrm{and}\\
&|y| \equiv \frac{1}{2}|x| + \langle z, x \rangle \pmod{2}.
\end{align}
\end{problem}

As discussed, we can use other graphs instead of the grid to define different variants of this problem. 
For instance, a balanced binary tree has lower diameter, which leads to slightly better parameters, but we use the grid to achieve a spatially local quantum circuit. 

A path graph (sometimes called the line graph) is even simpler than the grid or tree. Unfortunately, the Relaxed Parity Halving Problem corresponding to the path graph can be solved by an $\NC^{0}$ circuit 
(see \cref{footnote:line}), which makes it unsuitable for proving a separation against $\NC^0$. 

\subsection{Quantum circuit and \texorpdfstring{$\AC^0$}{AC0} lower bound}

In this section we establish \Cref{thm:GRPHP}, which states that Grid-RPHP can be solved in $\QNC^0$, but it is average-case hard for $\AC^0$ circuits. We start by establishing the quantum upper bound.

\begin{thm}[Grid-RPHP is in $\QNC^0$]
\label{thm:GRPHPupper}
There exists a depth-$5$ spatially local $\QNC^{0}$ circuit that exactly solves $\GRPHP$.
\end{thm}
\begin{proof}
Let $G = (V,E)$ be a $O(\sqrt{n})$-diameter spanning tree of the $\sqrt{n} \times \sqrt{n}$ grid graph with $|V| = n$ and $|E|=n-1$. This graph has degree $\Delta=3$.
As shown in \Cref{thm:PMCatState}, there is a spatially local $\QNC^{0}$ circuit of depth $\Delta+2$ to construct a random poor man's cat state $\frac{\ket{z} + \ket{\overline{z}}}{\sqrt{2}}$ (for some $z \in \B^{V}$) and the associated string $d \in \B^{E}$ for graph $G$ such that 
\begin{equation}
d_{(u,v)} = z_u \oplus z_v
\end{equation} 
for all $(u,v) \in E$. We run the $\QNC^0$ circuit for the Parity Halving Problem as per \Cref{thm:phpbutpoorer}, and get an output $y \in \B^{V}$ such that 
\begin{equation}
|y| \equiv \frac{1}{2}|x| + \langle z, x \rangle \pmod{2}.
\end{equation}
We have defined $\GRPHP$ so that it is not necessary to compute $z$, just a vector $d$ consistent $z$, which we have from the construction of the poor man's cat state. Hence, we return $y$ and $d$ satisfying the condition, and we are done. 
\end{proof}

This result is not surprising, since the relaxed parity halving problem is a relaxation of the parity halving problem explicitly constructed with the goal of having a $\QNC^0$ circuit.

The nontrivial direction of the argument is to show that $\AC^{0}$ cannot solve $\GRPHP$. We accomplish this by exhibiting an $\NC^0$ reduction to  $\GRPHP$ from an instance of $\PHP_{n,m}$ with $m = \Theta(n^{3/2})$, which is still hard for $\AC^0$. Note that although our reduction is an $\NC^0$ reduction, it cannot be carried out with a $\QNC^0$ circuit, so we are not showing that $\PHP_{n,m}$ is in $\QNC^0$. 
While this might seem mysterious, the reason is that $\NC^0$ has one ability that we have not given $\QNC^0$: unbounded fan-out. Our reduction uses the fact that $\NC^0$ can make unlimited copies of the output of a gate, whereas $\QNC^0$ cannot do so.

\begin{thm}[Grid-RPHP is not in $\AC^0$]
\label{thm:reduction}\label{thm:GRPHPlower}
There is an $\NC^{0}$ reduction from $\PHP_{n,O(n^{3/2})}$ to $\GRPHP_n$. In particular, an $\AC^{0}$ circuit of size $s \le \exp(n^{1/2(d+1)})$ and depth $d$ cannot solve $\GRPHP_n$ with probability better than 
\begin{equation}
\frac{1}{2} + \exp(-n^{1/2-o(1)}\Big/O(\log s)^{2d})
\end{equation}
on a uniformly random input with even parity.
\end{thm}
\begin{proof}
Suppose we want to solve an instance of $\PHP_{n, O(n^{3/2})}$, and we can solve $\GRPHP_n$. Let $T = (V,E)$ be a $O(\sqrt{n})$-diameter spanning tree of an $n$ vertex grid graph.

Take the input $x \in \B^{n}$ from the $\PHP$ instance as input for a $\GRPHP_n$ instance, mapping the $n$ bits arbitrarily to vertices of the grid. Solving this $\GRPHP$ instance gives us two outputs $y \in \B^{V}$ and $d \in \B^{E}$ such that 
\begin{equation}
|y| \equiv|x|/2 + \langle z, x \rangle \pmod{2}, 
\end{equation}
where $z \in \B^{V}$ satisfies the parity constraints in $d$. In particular, if we fix $z_1 = 0$ then each $z_i$ is the parity of all $d_j$ along a path from $z_1$ to $z_i$ in the graph. Let $D_i \subseteq |E|$ denote the edges in the path from $z_1$ to $z_i$. Then we can write 
\begin{align}
\langle z, x \rangle &= \sum_{i} z_i x_i \\
&= \sum_{i} \sum_{j \in D_i} d_j x_i.
\end{align}
Since the diameter of the grid graph is $O(\sqrt{n})$, we may assume each $D_i$ has size at most $O(\sqrt{n})$. Thus, we have expressed $\langle z, x \rangle$ as a sum of $O(n^{3/2})$ terms of the form $d_j x_i$. Note that any such term $d_j x_i$ is easy to compute with a single $\mathrm{AND}$ gate, since we have the string $x$ (our input) and string $d$ (the output of the $\GRPHP$ circuit) available.

We now create $O(n^{3/2})$ new output bits, one for each term $d_j x_i$ that appears in the sum above. If we call this string of length $O(n^{3/2})$ $y'$, then our final output for $\PHP_{n,O(n^{3/2})}$ is the string $y$ (the output of the $\GRPHP$ circuit) concatenated with the string $y'$. We claim this is a correct solution to our $\PHP_{n,O(n^{3/2})}$ instance. This is because $|y'|=\langle z,x \rangle$ by construction, and hence the output to $\PHP_{n,O(n^{3/2})}$, which is the concatenated string $(y,y')$, has parity 
\begin{equation}
|y|+|y'| \equiv |x|/2 + \langle z,x \rangle + |y'| \equiv |x|/2 \pmod 2,
\end{equation}
which satisfies the output condition of $\PHP_{n,O(n^{3/2})}$.

Note that using this reduction, if $\GRPHP$ is solved by an $\AC^0$ circuit of size $s$ and depth $d$, then we get an $\AC^0$ circuit for $\PHP_{n,O(n^{3/2})}$ of size $s+O(n^{3/2})$ and depth $d+1$. Applying \Cref{thm:AC0} gives the required bound assuming $s \le \exp(n^{1/2(d+1)})$.
\end{proof}

\section{Parallel Grid-RPHP}

A common way to decrease the success probability of a problem is to repeat it in parallel and require success on every instance. We define the Parallel version of $\GRPHP$ to simply be some number of copies of $\GRPHP$ where the correctness condition is that all outputs must be correct. Obviously if there is a $\QNC^0$ circuit for $\GRPHP$ then there is one for $\PGRPHP$. So the $\QNC^0$ upper bound in  \Cref{thm:PGRPHP} is clearly true. The remainder of this section is devoted to proving the lower bound in \Cref{thm:PGRPHP}.

We need to show that $\PGRPHP$ becomes harder for $\AC^0$ circuits, but let us start with showing that a parallel version of the Parity Halving Problem gets harder with more copies, and then we will reduce $\PGRPHP$ to this. Note that although the copies of the game are played in parallel, this does not represent a so-called ``parallel repetition result" for $\GRPHP$ because different copies of the game are played by \emph{different} players.

\subsection{Parallel Parity Halving Problem}

We now define the parallel version of PHP, with $k$ copies of the problem.

\begin{problem}[Parallel Parity Halving Problem, $\PHP_{n,m}^{\otimes k}$]
Given $k$ strings $x_1, \ldots, x_k \in \{ 0, 1 \}^{n}$ of length $n$ as input, promised that each $x_i$ has even parity, output $k$ strings $y_1, \ldots, y_k \in \{ 0, 1 \}^{m}$ of length $m$ such that 
\begin{equation}
|y_i|  \equiv \frac{1}{2} |x_i| \pmod{2}
\end{equation}
for all $1 \leq i \leq k$. 
\end{problem}

In other words, $\PHP_{n,m}^{\otimes k}$ is simply $k$ independent copies of the Parity Halving Problem, and the players win if they solve all of the subgames simultaneously. Clearly a $\QNC^0$ circuit will have no problem solving this, given $k$ cat states.

\begin{prop}
There is a depth-$2$, linear-size $\QNC^{0}/\qpoly$ circuit which solves $\PHP_{n,n}^{\otimes k}$ with certainty. More specifically, the quantum advice is $k$ cat states of size $n$, $|\Cat_n\>^{\otimes k}$.
\end{prop}

The classical lower bound, however, will require some new ideas. It is not always easy to show that solving many independent instances of a problem is as hard as solving all of them independently, because of the possibility of correlating success in one instance with success in another. We use Vazirani's XOR Lemma~\cite{Vaz86} to attack this problem indirectly. 

\begin{lem}[Vazirani's XOR Lemma]
\label{lem:vazirani}
Let $D$ be a distribution on $\mathbb F_2^{m}$ and $p_S$ denote the parity function on the set $S\subseteq [m]$, defined as $p_S(x)=\oplus_{i\in S} \, x_i$. If $|\mathbb{E}_{x \in D}[(-1)^{p_{S}(x)}]| \leq \varepsilon$ for every non-empty subset $S \subseteq [m]$, then $D$ is $\varepsilon \cdot 2^{m/2}$ close (in statistical distance) to the uniform distribution over $\mathbb F_2^{m}$. 
\end{lem}

Note that this bound is very intuitive when $\eps=0$. It says that if a distribution has the property that on every subset of bits, if the induced distribution places equal mass on even and odd parity strings, then this distribution must be the uniform distribution. 

To get an effective bound in \Cref{lem:vazirani} we need to guarantee that $\eps < 2^{-m/2}$.
The following simple lemma handles bigger $\eps$ effectively, but only guarantees that the probability to sample the all zeros input is small, as opposed to guaranteeing that the distribution is close to uniform.

\begin{lem}[Special case of the XOR Lemma]
\label{lem:vazirani2}
Let $D$ be a distribution on $\mathbb F_2^{m}$.
 If $|\mathbb{E}_{x \in D}[(-1)^{p_{S}(x)}]| \leq \varepsilon$ for every non-empty subset $S \subseteq [m]$, then, 
 $\Pr_{x\sim D}[ x= 0^m] \le  2^{-m} +  \eps.$
\end{lem}
\begin{proof}
Fix $y\in \mathbb F_2^m$.
Let $f:\B^m \to \B$ be the indicator function that checks whether a given input is equal to $0^m$.
\begin{equation}
f(x) = \prod_{i=1}^{m} \frac{(-1)^{x_i} + 1}{2} = 2^{-m} \sum_{S \subseteq [m]}(-1)^{p_S(x)} = 2^{-m} + 2^{-m} \sum_{\emptyset \neq S\subseteq [m]} (-1)^{p_{S}(x)}
\end{equation}
Thus,
\begin{equation}
	\Pr_{x\sim D}[x=0^m] = \mathbb{E}_{x\sim D}[f(x)] 
	\le  2^{-m} +  2^{-m} \cdot \sum_{\emptyset \neq S\subseteq [m]} \left|\mathbb{E}_{x\sim D}[(-1)^{p_{S}(x)}]\right| \le 2^{-m} + \eps.\;\qedhere
\end{equation}
\end{proof}

The relevance of the XOR lemma is the following: Consider the task of solving $k$ instances of some problem with some class of circuits. Say we know that solving $1$ instance of the problem is hard, in the sense that no circuit from our class solves the problem with probability significantly greater than half on some hard distribution over inputs.
For our task with $k$ instances we will choose the input distribution to be this hard distribution on all instances independently.
Now define a single bit random variable for each instance that indicates whether a given circuit correctly solved that instance on our chosen distribution. 
We know that for $1$ instance this bit, the random variable we defined, is essentially a coin flip. 
If we can prove that each bit is essentially a coin flip, and furthermore that the XOR of any subset of bits is essentially a coin flip, then we will get that the distribution is essentially uniform. 
Which means the probability of getting the all zeros output, which corresponds to the circuit correctly solving all instances, is exponentially small.

Consider an instance of $\PHP_{n,m}^{\otimes k}$. 
If we solve all the instances correctly, then the parity of the entire output of length $km$ (all instances included) is the same as half the entire input's Hamming weight mod $2$. This just follows from the definition of $\PHP$.
If we solve all but one instance correctly, then this condition will not hold. 
In general, we fail on an even number of subgames if the parity of the entire output is the same as half the entire input Hamming weight mod $2$. 
But that is just the usual condition for the Parity Halving Problem on an input of size $kn$ and output of size $km$. 
The only difference is that each instance additionally has an even-parity input. 
Thus we need a stronger version of \Cref{thm:parityhalvinggame} which allows for parity constraints on the input. As in \Cref{thm:parityhalvinggame}, we prove this theorem in the language of non-local games.

\begin{thm}
\label{thm:strongPHG}
Consider the \emph{constrained} Parity Halving Game on $n$ players, in which the inputs of $d_1$ players are fixed, and the remaining $n-d_1$ players are partitioned into $d_2$ parts (of size $\geq 2$) with each part constrained to some fixed parity. The probability of winning this version of the problem is 
\begin{equation}
\Pr[\mathrm{Win}] \leq \frac{1}{2} + 2^{-(n-d_1)/2 + d_2}. 
\end{equation}
\end{thm}
\begin{proof}
The proof builds on \Cref{thm:parityhalvinggame}. The main change is that we need a different function $f$, to capture the different promise. Suppose for now that $d_1 = 0$. Say the input bits are divided into sets $S_1, \ldots, S_d$ and constrained to have parity $p_1, \ldots, p_d \in \{ 0, 1 \}$. We define $f$ such that 
\begin{equation}
f(x) := \frac{1}{2^d} \prod_{k=1}^{d} \left( i^{\sum S_k} + (-1)^{p_d} (-i)^{\sum S_k} \right)
\end{equation}
The idea is that $\frac{1}{2} \left( i^{\sum S_k} + (-i)^{\sum S_k} \right)$ is exactly $\Re(i^{\sum S_k})$, so it is $0$ if the parity on $S_k$ is odd and $i^{\sum S_k}$ otherwise. Similarly, $\frac{1}{2} \left( i^{\sum S_k} - (-i)^{\sum S_k} \right)$ is $0$ if the parity on $S_k$ is even and $i^{\sum S_k}$ otherwise. Altogether, this means that $f(x)$ is $0$ if the promise is violated and $i^{|x|}$ otherwise.

On the other hand, if we expand $f(x)$, we see that it is a convex combination of terms of the form $\pm (\pm i)^{x_1} (\pm i)^{x_2} \cdots (\pm i)^{x_n}$, and we have essentially already argued that 
\begin{equation}
\left| \sum_{x} (-1)^{a + b \cdot x} (\pm i)^{x_1} \cdots (\pm i)^{x_n} \right| \leq 2^{n/2}.
\end{equation}
It follows that the correlation of $(-1)^{a + b \cdot x}$ and $f(x)$, denoted by $\chi$, is at most $\frac{2^{n/2}}{2^{n-d}} = 2^{-n/2+d}$.

The probability of winning the game on a random input satisfying the promise is $\frac{1 + \chi}{2}$, or 
\begin{equation}
\Pr[\textrm{Win}] \leq \frac{1}{2} + 2^{-n/2 + d - 1}.
\end{equation}
Plugging $d_1$ and $d_2$ back in, we have 
\begin{equation}
\Pr[\textrm{Win}] \leq \frac{1}{2} + 2^{-(n-d_1)/2 + d_2 - 1},
\end{equation}
as desired.
\end{proof}

With this result we can now show that the Parallel Parity Halving Problem is hard for $\AC^0$ circuits.

\begin{thm}
[Parallel PHP is not in $\AC^0$]
\label{thm:Parallel PHP}
Let $k=n$ and $m \in [n, n^2]$.
Any $\AC^{0}$ circuit $F: \B^{nk} \to \B^{mk}$ with depth $d$ and size $s \le \exp((kn)^{1/2d})$ solves $\PHP_{n,m}^{\otimes k}$ with probability at most $\exp(n^2 /(m^{1+o(1)} \cdot O(\log s)^{2(d-1)}))$.
\end{thm}
\begin{proof}
Much like \Cref{thm:AC0}, we assume $F$ solves the problem and randomly restrict it. We pick parameters $q = \sqrt{\log (mk)}$,  $p = 1/(O(\log s)^{d-1} \cdot (mk)^{1/q})$ that are similar to the ones picked in \Cref{thm:AC0}, but adjusted to the input and output lengths. We apply $p$-random restrictions. Most of the time this will simplify $F$ to the point that $F|_{\rho} \in \DT(pnk/4) \circ \DT(q-1)^{mk}$ and $\rho$ keeps at least $pnk/2$ variables alive. However, with probability $\exp(-\Omega(pnk))$ the restriction will fail, and we are forced to assume (pessimistically) that $F|_{\rho}$ solves the problem perfectly in these cases. This probability of failure is negligible compared to $\exp(-n^2 /(m^{1+o(1)} \cdot O(\log s)^{d-1}))$.

Let us assume the random restriction did its job, and now $F|_{\rho}$ is computed by common partial decision tree followed by a forest of depth-$q$ decision trees. 
By averaging argument, it suffices to show that for each leaf $\lambda$ in the partial decision tree, $F_{\rho,\lambda}$ solves the Parallel Parity Halving Problem on the legal inputs consistent with  $(\rho, \lambda)$ with exponentially small probability.

For each leaf in the partial decision tree, $\lambda$, the circuit $F_{\rho,\lambda}$ has $pnk/4$ inputs, $mk$ outputs, and locality $2^q$. By \Cref{prop:independentinputs}, we may further restrict down to a subset of 
$\Omega((p n k)^2 / (m k 2^{2q}))$ inputs so that each output bit depends on at most one input bit. Finally, we restrict one more time to eliminate subgames where fewer than say, half the average number of inputs (which is $\Omega((pn)^2 /(m 2^{2q})$) are unrestricted.  Subgames with too few inputs may be too easy to win and prevent us from using the XOR Lemma.  This last step kills at most half of the input bits that were alive before this step.

The remaining circuit, which we will call $C$, satisfies the following:
\begin{itemize}
    \item $C$ has $\Omega((p n)^2 k / (m 2^{2q}))$ unrestricted inputs.
    \item Each output bit depends on at most one input bit. 
    \item Each subgame is either fixed or has at least $n' = \Omega((pn)^2 /(m 2^{2q}))$ unrestricted bits.
\end{itemize}
Note that there are at least $\ell = \Omega(p^2 n k / (m 2^{2q})) = \Omega(n')$ remaining subgames.
We shall show that the probability to win the remaining subgames is at most $2^{-\Omega(n')}$.

Define a distribution $D$ on $\ell$ bit strings which runs the circuit $C$ on a random input satisfying the promise, and outputs a string of bits, $w \in \{0,1\}^\ell$, one bit for each of the remaining $\ell$ subgames, which is $0$ if the circuit wins the subgame and $1$ if the circuit loses the subgame. We will argue that
\begin{equation}
\abs{\mathbb{E}_{w \in D}[(-1)^{p_S(w)}]} \leq 2^{-\Omega(n')}
\end{equation}
for each non-empty $S \subseteq [\ell]$. That is, the parity of any non-empty subset of the games is exponentially close to a coin flip. Once we show this, \Cref{lem:vazirani2} applies and says that the probability to sample $0^{\ell}$ from $D$ is at most $2^{-\ell} + 2^{-\Omega(n')} \le 2^{-\Omega(n')}$. Thus, we win with probability $2^{-\Omega(n')}$, which  dominates the $\exp(-\Omega(pnk))$ probability that the restriction fails.  

All that remains to show is that for any non-empty subset of  subgames, the circuit loses an even number of subgames with probability exponentially close to $\frac{1}{2}$. Notice that losing an even number of subgames is equivalent to winning a constrained Parity Halving game on the combined inputs and outputs of the subgames,  where the constraints ensure each subgame has even parity input.
For a subset $S \subseteq [\ell]$ of the subgames,
 \Cref{thm:strongPHG} says that the probability of winning is at most $\frac{1}{2} + 2^{-\Omega(|S| n')+|S|}$, since there are at least $|S| n'$ inputs and at most $|S|$ relevant constraints. This is maximized when $|S| = 1$, where we get that the probability of winning is at most $1/2 + 2^{-\Omega(n')}$. 
 
 To finish we note that $n' = \Omega((pn)^2 /(m 2^{2q}) \ge n^2 / (m^{1+o(1)} \cdot O(\log s)^{2(d-1)})$.
\end{proof}

\subsection{Parallel Grid-RPHP}

We are now ready to prove the lower bound for $\PGRPHP$ by reduction to \Cref{thm:Parallel PHP}. As before, we define Grid-$\RPHP_n^{\otimes k}$ to be the problem where we are given $k$ copies of Grid-$\RPHP_n$, and the correctness condition is that all copies are correct. We are now ready to show the lower bound.

\begin{thm}[Parallel $\GRPHP$ is not in $\AC^0$]
\label{thm:GridRPHPlb}
Choose $k=n$. Any $\AC^{0}$  circuit of size $s$ and depth $d$ for Grid-$\RPHP_n^{\otimes k}$ succeeds with probability at most $\exp(-n^{1/2-o(1)}/O(\log s)^{2d})$.
\end{thm}
\begin{proof}
We use the reduction from $\PHP_{n,O(n^{3/2})}$ to Grid-$\RPHP_n$ (\Cref{thm:reduction}) on each of the $k$ instances. Note that the reduction does not change the inputs, and only manipulates the output bits. 
Assume $C$ is a circuit of size $s$ and depth $d$ for Grid-$\RPHP_n^{\otimes k}$ that succeeds with probability $\eps$ (over the uniform distribution over inputs that satisfy the promise).
Then, there exists a circuit $C'$ of size $\poly(s)$ and depth $d+1$ that succeeds with probability at least $\eps$ on the same input distribution.
By \Cref{thm:Parallel PHP}, with $m = O(n^{3/2})$, we get that
\begin{equation}
    \eps \le \exp(-n^{1/2-o(1)}\Big/O(\log s)^{2d}),
\end{equation}
which yields the bound stated in the theorem.
\end{proof}

This implies the second part of \Cref{thm:PGRPHP} and completes its proof.

\section{Relation to Hidden Linear Function Problems}
\label{sec:HLF}

Finally, to establish our main result (\Cref{thm:main}), we have to show that $\PGRPHP$ reduces to the \DDHLF problem. Since we have already established the required hardness for $\PGRPHP$ in \Cref{thm:PGRPHP}, we will then be done.

We start by recalling the general Hidden Linear Function problem (\HLF) defined by Bravyi, Gosset, and K\"{o}nig \cite{BGK18}.
\begin{problem}[Hidden Linear Function problem]
We are given as input a symmetric matrix $A \in \mathbb \{ 0, 1 \}^{n \times n}$ and vector $b \in \mathbb \{ 0, 1, 2, 3 \}^n$. From these, define a quadratic form $q \colon \mathbb F_2^{n} \to \mathbb Z_4$ as $q(u) := u^{T} A u + b^{T} u \pmod 4$. Define $\mathcal{L}_q$ as follows:
\begin{equation}
\mathcal{L}_q := \{ u \in \mathbb F_2^{n} :\forall v \in \mathbb F_2^{n}, \, q(u \oplus v) \equiv q(u) + q(v) \pmod{4}\}.
\end{equation}
Bravyi et al.\ \cite{BGK18} show that (i) $\mathcal{L}_q$ is a linear subspace of $\mathbb F_2^{n}$, (ii) for all $u\in\mathcal{L}_q$, $q(u) \in \{ 0, 2 \}$, and (iii) $q$ is linear on $\mathcal{L}_q$.
Since $q$ is linear on $\mathcal{L}_q$, there exists a $p\in  \mathbb F_2^{n}$ such that $q(u) \equiv 2p^T u \pmod 4$ for all $u\in\mathcal{L}_q$. The goal is to output any string $p$ satisfying this condition.
\end{problem}

\begin{thm}
\label{thm:hlfreduction}
There is an $\NC^{0}$ reduction from $\RPHP$ on any graph $G$ to the \HLF problem.
\end{thm}
\begin{proof}
As discussed, the Relaxed Parity Halving Problem is well defined for any connected graph $G = (V,E)$ (\Cref{problem:RPHP}). Given an even-parity vector $x \in \{ 0, 1 \}^{V}$, the goal is to output $y \in \{ 0, 1 \}^{V}$ and $d \in \{ 0, 1 \}^{E}$, such that the parity of $d$ on any cycle is 0, and the unique (up to complement) $z \in \{ 0, 1 \}^{V}$ satisfying the parity conditions implied by $d$ also satisfies
\begin{equation}
|y| \equiv\frac{|x|}{2} + z^{T} x \pmod{2}.
\end{equation}
Note that if we take the complement of $z$ instead, the condition does not change because $(z^{T} + \overline{z}^{T}) x \equiv 1^{T} x \equiv |x| \equiv 0 \pmod{2}$. 

We now describe our reduction from $\RPHP$ on $G = (V,E)$ to the HLF problem. In our reduction, $n=|V|+|E|$. We define the input $A \in \mathbb \{ 0, 1 \}^{n \times n}$ to HLF from the graph $G$, and the input $b\in \mathbb \{ 0, 1, 2, 3 \}^n$ to HLF from the input $x\in \{ 0, 1 \}^{V}$ to $\RPHP$.  Let $M \in \mathbb F_2^{|V| \times |E|}$ be the incidence matrix of the graph $G$. Then we define $A$ and $b$ as
\begin{equation}
A = \begin{pmatrix}
0 & M \\
M^{T} & 0
\end{pmatrix}, \  \mathrm{and} \   
b = \begin{pmatrix}
x\\
0
\end{pmatrix},
\end{equation}
where the $0$ above refers to the all zeros matrix or vector as appropriate. The solution to this HLF problem is some vector $p$, which we claim solves the Relaxed Parity Halving Problem with the identification $p= \begin{pmatrix} y\\ d \end{pmatrix}$. Let us verify this claim.

With this choice of $A$ and $b$, the quadratic form $q$ becomes 
\begin{align}
q \begin{pmatrix} u_V \\ u_E \end{pmatrix} &\equiv u_V^{T} M u_E + u_E^{T} M^{T} u_V + x^{T} u_V  \pmod{4} \\
&\equiv 2 u_V^{T} M u_E + u_V^{T} x \pmod{4}.\label{eq:quadratic-form}
\end{align}

Now note that since $G$ is connected, $M$ has rank $|V|-1$. Specifically, the column span $\{ Mw : w \in \mathbb F_2^{|E|} \} \subseteq \mathbb F_2^{|V|}$ is the set of all vectors of even parity. 
Since our input to RPHP, $x \in \mathbb F_2^{|V|}$, has even parity, there exists a $w \in \mathbb F_2^{|E|}$ such that $x = Mw$. Let us show that the vector $u$ is in $\mathcal{L}_q$, where 
\begin{equation}
u := \begin{pmatrix}
1 \\ w 
\end{pmatrix},
\end{equation}
and $1$ is the all ones vector of size $|V|$. This means we want that for all $v$, $q(u\oplus v) - q(u) - q(v) \equiv 0 \pmod{4}$. Let us verify this  for an arbitrary vector $v = \begin{pmatrix} v_V \\ v_E \end{pmatrix}$ using the following calculation (performed modulo 4):
\begin{align}
& q\left( \begin{pmatrix} 1 \\ w \end{pmatrix} \oplus \begin{pmatrix} v_{V} \\ v_{E} \end{pmatrix} \right) - q \begin{pmatrix} 1 \\ w \end{pmatrix} - q \begin{pmatrix} v_V \\ v_{E} \end{pmatrix} \\
&\equiv 2 (1 \oplus v_V)^{T} M (w \oplus v_E) + (1 \oplus v_V)^{T} x - 2 (1^{T} M w) - 1^{T} x - 2 v_V^{T} M v_E - v_V^{T} x \\
&\equiv 2 (1 + v_V)^{T} M (w + v_E) + (1 \oplus v_V)^{T} x - 2 (1^{T} M w) - 1^{T} x - 2 v_V^{T} M v_E - v_V^{T} x \label{eq:ab} \\
&\equiv 2 (1^{T} M v_E) + 2 v_V^{T} M w + (1 \oplus v_V)^{T} x - 1^{T} x - v_V^{T} x \\
&\equiv 2 v_V^{T} x + (1 \oplus v_V)^{T} x - 1^{T} x - v_V^{T} x \label{eq:Mwx}\\
&\equiv 2 v_V^{T} x + (1 - v_V)^{T} x - 1^{T} x - v_V^{T} x \label{eq:oneminus}\\
&\equiv 0.
\end{align}
In \cref{eq:ab}, we used that for $a,b\in \B$, we have $2(a\oplus b)\equiv2(a+b) \pmod{4}$. In \cref{eq:Mwx} we used $Mw \equiv x \pmod{2}$ and that $1^{T} M \equiv 0 \pmod{2}$, since $M$ is an incidence matrix and any column has exactly two ones. In \cref{eq:oneminus} we used that for any $a\in\B$, $1\oplus a = 1-a$.

Since $u \in \mathcal{L}_q$, and $q$ is linear on $\mathcal{L}_q$, we have 
\begin{align}\label{eq:qu}
q(u) \equiv 2p^Tu &\equiv 2\begin{pmatrix} y \\ d \end{pmatrix}^T \begin{pmatrix} 1 \\ w \end{pmatrix} \equiv 2(y^{T} 1 + d^{T} w) \pmod{4}.
\end{align}

To prove that the output of HLF on our chosen inputs is a valid output for $\RPHP$, we need to verify the two conditions in \Cref{problem:RPHP}. The first condition requires $d$ to be related to some $z \in \B^V$; the entries of $d$ are differences (along edges) of two entries in $z$. If such a $z$ existed, then we would have $d = M^{T} z$. As noted after \Cref{problem:RPHP}, for $z$ to exist it suffices to show that the parity of $d$ along any cycle is $0$. In other words, we need to show that if $c \in \B^{E}$ is the indicator vector of any cycle, then $d^T c \equiv 0 \pmod{2}$.

To show this, note that for any cycle we have $Mc\equiv 0\pmod{2}$. Since $q ( \begin{smallmatrix} 0 \\ c \end{smallmatrix} ) = 0$, we also have
\begin{align}
q\left( \begin{pmatrix} u_V \\ u_E \end{pmatrix} \oplus \begin{pmatrix} 0 \\ c \end{pmatrix} \right) &\equiv 2 u_V^{T} M (u_E \oplus c) + u_V^{T} x \pmod{4} \\
&\equiv 2 u_V^{T} M u_E + 2 u_V^{T} M c + u_V^{T} x \pmod{4} \\
&\equiv q \begin{pmatrix} u_V \\ u_E \end{pmatrix} + q \begin{pmatrix} 0 \\ c \end{pmatrix} \pmod{4}.
\end{align}
In other words, $( \begin{smallmatrix} 0 \\ c \end{smallmatrix} )$ is in $\mathcal{L}_q$ for all cycles. Thus,
\begin{equation}
0 \equiv q \begin{pmatrix} 0 \\ c \end{pmatrix} \equiv 2 \begin{pmatrix} y \\ d \end{pmatrix}^T \begin{pmatrix} 0 \\ c \end{pmatrix} \equiv 2 d^{T} c \pmod{4},
\end{equation}
which implies that $d^T c \equiv 0 \pmod{2}$.

Now that we know $d = M^{T} z$ for some $z$, it follows that $d^{T} w = z^{T} M w = z^{T} x$, and therefore from \cref{eq:qu} we have $q(u) = 2(|y| + z^{T} x)$. On the other hand, we can also evaluate $q(u)$ from the quadratic form definition, \cref{eq:quadratic-form}, which gives
\begin{align}
q(u) &\equiv 2 (1^{T} M w) + 1^{T} x \pmod{4} \\
&\equiv 3|x| \pmod{4}\\
&\equiv |x| \pmod{4},
\end{align}
where the last equivalence follows from the fact that $x$ has even parity.

Thus solving the HLF problem on our chosen inputs produces a solution such that $2|y| + 2 z^{T} x \equiv |x| \pmod{4}$, which satisfies the second condition of \cref{problem:RPHP}. Hence the outputs to HLF, $y$ and $d$ (and implicitly $z$), satisfy the conditions of the Relaxed Parity Halving Problem.
\end{proof}

The main problem studied in Bravyi, Gosset and K\"{o}nig~\cite{BGK18} is actually a version of HLF on an $N \times N$ grid called the 2D Hidden Linear Function problem (\DDHLF). 
More specifically, in \DDHLF, the matrix $A$ is supported only on the grid in the sense that $A_{ij} = 0$ if there is no edge from vertex $i$ to vertex $j$ on the 2D grid. 
The reduction in \Cref{thm:hlfreduction} roughly preserves the topology of the graph, so it is not difficult to show a reduction from Relaxed Parity Halving on the 2D grid to the \DDHLF.
\begin{cor}\label{cor:Grid}
There is an $\NC^{0}$ reduction from $\GRPHP$ to $\DDHLF$.
\end{cor}
\begin{proof}
The plan is use the same reduction to HLF as above, and observe that the reduction actually creates a $\DDHLF$ instance. Obviously, $\RPHP$ on a grid starts from a grid graph. The reduction creates a matrix with $|V| + |E|$ rows and columns from this graph on $|V|$ vertices. 
We can think of the reduction as transforming the graph by creating a new vertex for each edge, then splitting each edge into two, both incident at the new vertex. The matrix $A$ is then supported on this transformed graph. Fortunately, the transformation takes the $n \times n$ grid graph to a subgraph of the $(2n-1) \times (2n-1)$ grid graph. This means the \HLF instance constructed when starting from a 2D graph is actually a \DDHLF instance, which completes the proof.
\end{proof}

Furthermore, we can solve multiple instances of $\GRPHP$ by solving a single instance of \DDHLF.
\begin{lem}\label{lem:ddhlf_sum}
Consider an instance of \HLF with input $A \in \{ 0, 1 \}^{n \times n}$ and $b \in \{ 0, 1, 2, 3 \}^{n}$ such that $A$ is block diagonal with blocks $A_1, \ldots, A_k$ (all square matrices of various sizes), and the corresponding division of $b$ is $b_1, \ldots, b_k$. Then any solution $p$ to the instance $(A,b)$ is a direct sum of solutions $p_i$ to instances $(A_i,b_i)$ of \HLF. 
\end{lem}
\begin{proof}
It suffices to prove the result for $k = 2$ blocks, since we can break up $k$ blocks into two blocks of size $1$ and $k-1$ and prove the claim by induction. 

Now let $u=(u_1,u_2)$. Then the block structure of $A$ gives $q(u) \equiv q_1(u_1) + q_2(u_2) \pmod{4}$, where 
\begin{align}
q_1(u_1) := u_1^{T} A_1 u_1 + b_1^{T} u_1 \pmod{4}, \quad \mathrm{and} \quad q_2(u_2) := u_2^{T} A_2 u_2 + b_2^{T} u_2 \pmod{4}.
\end{align}

Suppose we have $u = (u_1, u_2) \in \mathcal{L}_q$. By definition, for all $v = (v_1, v_2) \in \{ 0, 1 \}^{n}$ we have 
\begin{equation}
q(u_1 \oplus v_1, u_2 \oplus v_2) \equiv q(u_1, u_2) + q(v_1, v_2) \pmod{4}.    
\end{equation}
By rearranging the terms, we have
\begin{equation}
q_1(u_1 \oplus v_1) \equiv q_1(u_1) + q_1(v_1) + q_2(u_2) + q_2(v_2) - q_2(u_2 \oplus v_2) \pmod{4}.    
\end{equation}
In particular, if $v_2 = 0$ then for all $v_1$ we have
\begin{equation}
q_1(u_1 \oplus v_1) \equiv q_1(u_1) + q_1(v_1) \pmod{4},    
\end{equation}
which proves $u_1 \in \mathcal{L}_{q_1}$. Similarly, $u_2 \in \mathcal{L}_{q_2}$, and so we have that $\mathcal{L}_q = \mathcal{L}_{q_1} \oplus \mathcal{L}_{q_2}$.

Finally, suppose $p = (p_1, p_2)$ is a valid solution to HLF on input $(A, b)$. For all $u = (u_1, u_2) \in \mathcal{L}_q$ we have
\begin{equation}
q_1(u_1) + q_2(u_2) \equiv q(u) \equiv 2 p^{T} u \equiv 2p_1^{T} u_1 + 2p_2^{T} u_2 \pmod{4}.
\end{equation}
In particular, for all $u_1 \in \mathcal{L}_{q_1}$ we have $(u_1, 0) \in \mathcal{L}_q$ and hence 
\begin{equation}
q_1(u_1) \equiv 2 p_1^{T} u_1 \pmod{4}
\end{equation}
for all $u_1 \in \mathcal{L}_{q_1}$. It follows that $p_1$ is a solution to HLF on the instance $(A_1, b_1)$, and (by symmetry) $p_2$ is a solution to HLF on the instance $(A_2, b_2)$. 

Conversely, if $p_1$ and $p_2$ are solutions to $(A_1, b_1)$ and $(A_2, b_2)$, then $p=(p_1,p_2)$ is a solution to $(A,b)$ since
\begin{equation}
q(u) \equiv q_1(u_1) + q_2(u_2) \equiv 2p_1^{T} u_1 + 2p_2^{T} u_2 \equiv 2 p^{T} u \pmod{4},
\end{equation}
and $u \in \mathcal{L}_q$ implies $u_1 \in \mathcal{L}_{q_1}$ and $u_2 \in \mathcal{L}_{q_2}$.
\end{proof}

\begin{cor}
\label{cor:mainreduction}
Parallel $\GRPHP$ reduces to \DDHLF. 
\end{cor}
\begin{proof}
The parallel version of $\GRPHP$ is just many instances of $\RPHP$ which we are expected to solve simultaneously. If we reduce each instance to a \DDHLF problem, then combine the instances as in the lemma above, then solving the combined HLF instance gives solutions to all the individual HLF instances, and hence solutions for all the parallel $\GRPHP$ instances. 
\end{proof}

\section{Parity Bending Problem} 
\label{sec:PBP}

We now move on to the Parity Bending Problem discussed in the introduction.

\begin{problem}[Parity Bending Problem, $\PBP_n$] \label{problem: parity bending}
Given an input $x \in \{ 0, 1 \}^{n}$, output a string $y \in \{ 0, 1 \}^{n}$ such that 
\begin{align}
|y| &\equiv 0 \pmod{2} \;\; \text{ if } \;\; \abs{x} \equiv 0 \pmod{3} \text{ and }\\
\abs{y} &\equiv 1 \pmod{2} \; \; \text{ otherwise.}
\end{align}
\end{problem}

Our goal is to prove a quantum advantage for $\QNCcat$ circuits. 
\Cref{thm:PBP} states our main results about this problem.

\begin{repthm}{thm:PBP}
The Parity Bending Problem $(\PBP_n)$ can be solved by a depth-$2$, linear-size quantum circuit starting with the $|\Cat_n\>$ state with probability at least  $3/4$ on any input. 
But there exists an input distribution on which any $\AC^0[2]/\rpoly$ circuit of depth $d$ and size at most $\mathrm{exp}\bigl(n^{\frac{1}{10d}}\bigr)$ only solves the problem with probability $\frac{1}{2}+\frac{1}{n^{\Omega(1)}}$. 
\end{repthm}

Most of this section is devoted to a proof of \Cref{thm:PBP}. We begin by introducing the quantum circuit that solves this problem (\Cref{thm:PBG}) in \Cref{sec:upper}, and then prove the classical lower bound (\Cref{thm:AC02}) in \Cref{sec:lower}. Finally, in \Cref{sec:PPBP} we establish \Cref{thm:RPBP}, which shows that a parallel version of the game can further separate the success probabilities of classical and quantum circuits. 

\subsection{Upper bounds}
\label{sec:upper}

As in the previous section, the $\QNC^0$ circuit solving the Parity Bending Problem can be thought of as an implementation of a quantum strategy for a cooperative non-local game. In the Parity Bending Problem, the players are each given one bit of an $n$-bit string, and want to give an output satisfying the same criterion as in \Cref{problem: parity bending}. It is known that a quantum strategy can win this game with probability larger than classically possible, although even quantum players cannot achieve probability $1$~\cite{watts2018algorithms}. We now describe the quantum strategy.

\begin{thm}[Quantum circuit for $\PBP_n$] \label{thm:PBG}
There is a quantum strategy for the Parity Bending Problem which wins with certainty on inputs with Hamming weight  $0 \pmod 3$ and with probability $3/4$ on any other input. For an $n$-player game, this strategy only requires an $n$-qubit cat state, $|\Cat_n\>$. 
\end{thm}

\begin{proof}
The quantum strategy for this game is similar to the one for \Cref{prob:PHP} described in \Cref{thm:PHPquantum} and \Cref{fig:php}, except that the controlled-$S$ gates, which add a phase of $i$ if both qubits are $1$, are replaced with controlled-$R_z(2 \pi /3)$ gates, which add a phase of $e^{2\pi i/3}$ when both qubits are $1$. We define $R_z(2 \pi /3)$ to be the matrix $\left(\begin{smallmatrix} 1 & 0 \\ 0 & e^{2\pi i/3} \end{smallmatrix}\right)$.

We now describe the strategy in more detail. Each player starts with their input bit $x_i$ and a single qubit of the $|\Cat_n\>$ state. If their input bit is $x_i = 1$, they apply a rotation $R_z(2 \pi /3)$ to their qubit of the cat state, otherwise they do nothing. 
This is equivalent to applying a controlled-$R_z(2 \pi /3)$ gate with $x_i$ as the control qubit and their qubit of the cat state as the target.
Then they apply a Hadamard gate to their qubit of the cat state and measure that qubit. 

Given an input string $x$, after the controlled rotations and before the Hadamard gate, the cat state has been transformed into
\begin{align}
\bigotimes_{j=0}^n \bigl( R_z(2 \pi /3)^{x_j} \bigr) \frac{1}{\sqrt{2}}\bigl(\ket{0^n} + \ket{1^n}\bigr) &= \frac{1}{\sqrt{2}}\left( \ket{0^n} + \exp(\frac{2 \pi i \abs{x}}{3}) \ket{1^n} \right).
\end{align}

If $|x|\equiv 0\pmod 3$ this state is just the cat state. As noted in \Cref{sec:PHP}
\begin{align}
    H^{\otimes n} \left( \frac{1}{\sqrt{2}}\bigl(\ket{0^n} + \ket{1^n}\bigr) \right) = \ket{\Psi_\textrm{even}},
\end{align}
where $\ket{\Psi_\textrm{even}}$ is a uniform superposition over all even parity $n$-bit strings. So we see the players win the game with probability 1 on an input with Hamming weight $0 \pmod 3$. 

If $|x|\equiv 1 \pmod 3$ or $|x|\equiv 2 \pmod 3$, the state after rotation and before the Hadamard gates is given by
\begin{equation}
    \frac{1}{\sqrt{2}} \bigl(\ket{0^n} + \exp(\pm 2 \pi i /3) \ket{1^n}\bigr).
\end{equation}
Note that this state lives in the span of the states $\frac{1}{\sqrt{2}}\bigl(\ket{0^n} + \ket{1^n}\bigr)$ and $\frac{1}{\sqrt{2}}\bigl(\ket{0^n} - \ket{1^n}\bigr)$, and that 
\begin{align}
    H^{\otimes k} \left(\frac{1}{\sqrt{2}}\bigl(\ket{0^n} - \ket{1^n}\bigr) \right)= \ket{\Psi_\textrm{odd}},
\end{align}
with $\ket{\Psi_\textrm{odd}}$ being the uniform superposition over odd parity $n$-bit strings. Then the players win the game given input with Hamming weight 1 or $2\pmod 3$ with probability exactly
\begin{equation}
    \abs{\frac{1}{2}\bigl(\bra{0^n}-\bra{1^n}\bigr)\bigl(\ket{0^n} + \exp(\pm 2 \pi i /3) \ket{1^n} \bigr)}^2 = \frac{1}{4} (2 + 2\cos(\pm 2 \pi /3)) = \frac{3}{4}. 
    \qedhere
\end{equation}
\end{proof}
It is clear that the strategy described in \Cref{thm:PBG} can be implemented by a $\QNCcat$ circuit. 
By the analysis given above it is also clear that this circuit succeeds at Parity Bending with worst case probability $3/4$, and probability $5/6$ against inputs drawn from a uniform distribution.

\subsection{Lower bounds}
\label{sec:lower}

The main tool used in this section is a reduction from the Parity Bending Problem to the Mod 3 problem. We begin with a formal definition of the Mod 3 problem.

\begin{problem}[Mod 3] \label{problem:mod3}
Given an input $x \in \B^n$, output $y\in\{0,1\}$ such that $y=0$ if $|x|\equiv 0 \pmod {3}$, and $y=1$ otherwise.
\end{problem}

For our purposes, the key feature of Mod 3 problem is that it is hard to solve for $\AC^0[2]$ circuits, as shown by Smolensky~\cite{Smo87}: 

\begin{thm}[Mod 3 is not in $\AC^0{[2]}$]\label{thm:Smolensky}
Any $\AC^0[2]$ circuit of depth $d$ that computes the Mod 3 function (\Cref{problem:mod3}) must have size $\mathrm{exp}\bigl({\Omega(n^{\frac{1}{2d}})}\bigr)$. 
\end{thm}

From this worst-case lower bound it is not too hard to obtain an average-case lower bound.

\begin{lem}[Average-case lower bound for Mod 3]\label{lem:avgmod3}
There exists an input distribution on which any $\AC^0[2]/\rpoly$ circuit of depth $d$ and size at most $\mathrm{exp}\bigl(n^{\frac{1}{10d}}\bigr)$ only solves the Mod 3 function (\Cref{problem:mod3}) with probability $\frac{1}{2}+\frac{1}{n^{\Omega(1)}}$. 
\end{lem}
\begin{proof}
Toward a contradiction, assume that for all input distributions there exists an $\AC^0[2]/\rpoly$ circuit of depth $d$ and size $\mathrm{exp}\bigl(n^{\frac{1}{10d}}\bigr)$ that solves the Mod 3 problem with probability $1/2+\eps$ for $\eps = 1/n^{o(1)}$.

Then by Yao's minimax principle, there exists a probability distribution over $\AC^0[2]/\rpoly$ circuits, or equivalently over $\AC^0[2]$ circuits, that solves the Mod 3 problem with probability $1/2+\eps$ on every input. By sampling $O(1/\eps^2)$ $\AC^0[2]$ circuits from this probability distribution and taking the majority vote of their outcomes, we get a new $\AC^0[2]$ circuit that solves the Mod 3 problem with probability at least $0.99$ on every input. Now $1/\eps^2 = n^{o(1)}$, and it is easy to construct a depth-$d$ circuit of size $\exp(m^{O(1/d)})$ to compute the majority of $m$ variables~\cite{Has86}. Hence this majority circuit has depth $d$ and size  $\exp(n^{o(1/d)})$, which doubles the depth and does not significantly increase the size of our circuit.

Now we can amplify success probability $0.99$ to $1-\exp(-n)$ by again sampling $O(n)$ circuits that succeed with probability $0.99$ and taking their majority vote. This majority vote can be performed in $\AC^0$, since we only need to perform an approximate majority as constructed by Ajtai and Ben-Or~\cite{AB84}. 

Since this distribution over $\AC^0[2]$ circuits fails with probability less than $2^{-n}$, there exists one circuit in the distribution that works for all inputs. This yields an $\AC^0[2]$ circuit 
of depth $2d+O(1)$ and size $\mathrm{exp}\bigl(n^{\frac{1}{10d}}\bigr)$
computing Mod 3, which contradicts \Cref{thm:Smolensky}.    
\end{proof}

We can finally use this average-case bound to show our lower bound for the Parity Bending Problem.

\begin{thm}[$\PBP$ is not in $\AC^0{[2]}$]\label{thm:AC02}
There exists an input distribution on which any $\AC^0[2]/\rpoly$ circuit of depth $d$ and size at most $\mathrm{exp}\bigl(n^{\frac{1}{10d}}\bigr)$ only solves $\PBP_n$ with probability $\frac{1}{2}+\frac{1}{n^{\Omega(1)}}$. 
\end{thm}
\begin{proof}
Observe that any $\AC^0[2]$ circuit which solves the Parity Bending Problem with some probability can be extended to one that solves the Mod 3 problem with the same probability by adding a final parity gate over the output of the original circuit. The result then follows from \Cref{lem:avgmod3}.
\end{proof}

\subsection{Parallel Parity Bending Problem}\label{sec:PPBP}

Now our goal is to establish \Cref{thm:RPBP}, which strengthens the separation shown above. In this section we will consider  input and output strings over the ternary alphabet $\{0,1,2\}$.
Since we are talking about Boolean circuits manipulating these symbols, we actually mean that we encode these trits in binary using some canonical encoding, e.g., $\{0\mapsto 00, 1\mapsto 01, 2\mapsto 11\}$.
For a vector $x\in \{0,1,2\}^n$ we denote by $|x| = \sum_{i=1}^{n} {x_i}$.

To simplify our lower bounds, we modify the Parity Bending Problem to accept inputs drawn from $\{0,1,2\}^n$ when we move to the parallel version of the problem.

\begin{problem}[$k$-Parallel Parity Bending Problem] \label{problem:RPBP}
Given inputs $x_1, \ldots , x_k$ with $x_i \in \{0,1,2\}^n$ for all $i \in [k]$,  produce outputs $y_1, \ldots y_k \in \B^n$ such that $y_i$ satisfies:
\begin{align}
    \abs{y_i} &\equiv 0 \pmod{2} \; \text{  and  } \; \abs{x_i} \equiv 0 \pmod{3} \; \text{ or } \\
    \abs{y_i} &\equiv 1 \pmod{2} \; \text{ and } \; \abs{x_i} \not\equiv 0 \pmod{3} 
\end{align}
for at least a $\frac{2}{3} + 0.05$ fraction of the $i \in [k]$.
\end{problem}

This problem is $k$ copies of a problem very similar to $\PBP$, except that the inputs are now in $\{0,1,2\}$ instead of being in $\B$. But the quantum algorithm described in \Cref{sec:upper} can be easily modified to work in this case, by applying a controlled gate that applies the phase $\exp(2\pi i/3)$ when the control and target are $11$, and applies the phase $\exp(4\pi i/3)$ when the control and target are $21$.
By using this strategy for each individual gate, we get a $\QNC^0/\qpoly$ circuit which solves 
any given instance of the Parity Bending Problem with probability at least $3/4$. 
Since the Parallel Parity Bending Problem only requires $2/3+0.05$ of the instances to be solved correctly, by using this quantum strategy for each instance independently, the quantum circuit
solves the Parallel Parity Bending Problem with probability $1 - o(1)$. This establishes the quantum upper bound in \Cref{thm:RPBP}.

We now show that this problem is hard for $\AC^0[2]/\rpoly$ circuits.
We start by introducing a related problem:

\begin{problem}[3 Output Mod 3] \label{problem:3outMod3}
Given an input $x \in \{0,1,2\}^n$, output a trit $y \in \{0,1,2\}$ such that $y \equiv |x| \pmod {3}$.
\end{problem}

As expected, an $\AC^{0}[2]$ circuit cannot solve this problem. In fact, on the uniform distribution, an $\AC^{0}[2]$ circuit succeeds with probability close to $1/3$, which is trivially achieved by a circuit that just outputs $0$. The easiest way to see this is by using random self-reducibility.

\begin{lem}[Worst case to average case]
\label{lem:mod3selfreduce}
Suppose there is an $\AC^{0}[2]/\rpoly$ circuit $C$ of size $S$ and depth $d$ that solves \Cref{problem:3outMod3} on a uniformly random input with probability $1/3 + \varepsilon$ for some $\varepsilon$. That is,
\begin{equation}
\Pr_{x \in \{0,1,2\}^{n}} \left[C(x) \equiv |x| \pmod{3}\right] = \frac{1}{3} + \varepsilon.    
\end{equation}
Then there exists an $\AC^{0}[2]/\rpoly$ circuit $C'$ of depth $d+O(1)$ and size $S+O(n)$ such that for any $x \in \{0,1,2\}^{n}$, 
\begin{align}
&\Pr[C'(x) \equiv |x| \pmod{3}] = \frac{1}{3} + \eps, \quad \mathrm{and} \\    
&\Pr[C'(x) \equiv |x| + 1 \pmod{3}] = \Pr[C'(x) \equiv |x|+2 \pmod{3}] = \frac{1}{3} - \frac{\varepsilon}{2}.
\end{align}
\end{lem}
\begin{proof}
Although $\AC^0$ circuits cannot compute $|x| \bmod 3$ for an input string $x \in \{ 0, 1 ,2 \}^{n}$, it is possible for an $\AC^0$ circuit to sample a uniformly random vector $b \in \{ 0, 1 ,2 \}^{n}$ and its Hamming weight mod 3, $|b| \bmod{3}$, as follows.

Sample random trits\footnote{Technically, an $\AC^0$ cannot sample a trit with probability exactly $1/3$, but the probability can be made exponentially close to $1/3$.} $c \in \{ 0, 1, 2 \}^{n}$ and set $b_i = c_{i+1} - c_i$ for all $1 \leq i \leq n-1$ and $b_n = -c_n$. 
We claim that $b$ is a uniformly random string over $\{ 0, 1, 2 \}^{n}$. To see this, observe that $b_n$ equals a uniformly random trit, and furthermore for every $i \in [n-1]$, $b_i$  is uniformly random over $\{ 0, 1, 2 \}$, even conditioned on $c_{i+1}, \ldots, c_n$, and therefore $b_i$ is also uniformly random conditioned on $b_{i+1}, \ldots, b_{n}$. It follows that $b$ is a uniformly random string over $\{ 0, 1, 2 \}^{n}$. However, the advantage of this method of producing a random string (as opposed to simply sampling a random string) is that we know $|b|$ since 
\begin{equation}
|b| \equiv \sum_{i=1}^{n} b_i \equiv \sum_{i=1}^{n-1} (c_{i+1} - c_i) - c_n \equiv -c_1 \pmod{3}.    
\end{equation}

Now we use a random self-reduction to construct the claimed circuit $C'$. The circuit $C'$ first  samples a random $a \in \{ 1, -1 \}$ and a random $b \in \{ 0, 1, 2 \}^{n}$ with known Hamming weight $|b|$, as described above. Then let $C'$ output 
\begin{equation}\label{eq:modarith}
C'(x) := \frac{C(a\cdot x + b \bmod{3}) - |b|}{a} \pmod{3}.     
\end{equation}
It is clear that $|a \cdot x + b| \equiv a |x| + |b| \bmod{3}$, and that $a \cdot x + b$ is uniformly random regardless of $a$ and $x$, so we have 
\begin{align}
\Pr_{a,b}[C'(x) \equiv |x| + k \pmod{3}] &= \Pr_{a,b}[C(a \cdot x + b) \equiv a |x| + |b| + ak \pmod 3] \\
&= \Pr_{a,b}[C(a \cdot x + b) \equiv |a \cdot x + b| + ak \pmod 3] \\
&= \Pr_{y \in \{0,1,2\}^{n},a}[C(y) \equiv |y| + ak \pmod 3].
\end{align}
In particular, when $k = 0$ we have 
\begin{equation}
\Pr_{a,b}[C'(x) \equiv |x| \pmod{3}] = \Pr_y[C(y) \equiv |y| \pmod{3}] = \frac{1}{3} + \varepsilon.    
\end{equation}
When $k \neq 0$, we observe that $ak \neq 0$ is uniformly random and independent of $y$, so 
\begin{equation}
    \Pr[C'(x) \equiv |x| + 1 \pmod 3] = \Pr[C'(x) \equiv |x| + 2 \pmod 3]= \frac{1}{3} - \frac{\varepsilon}{2}.
\end{equation}

Observe that the sampling circuit above is of constant depth and linear size. The modular arithmetic performed in \cref{eq:modarith} can be performed by a constant-sized circuit. Overall this reduction increases the depth by a constant and the size by $O(n)$.
\end{proof}

\begin{lem}[Average-case lower bound]
\label{lem:3outMod3}
An $\AC^0[2]/\rpoly$ circuit of depth $d$ and size at most $\mathrm{exp}\bigl(n^{\frac{1}{10d}}\bigr)$ solves \Cref{problem:3outMod3} on a uniform distribution with probability at most $\frac{1}{3} + \frac{1}{n^{\Omega(1)}}$. 
\end{lem}
\begin{proof}
Let $A$ be the circuit that solves \Cref{problem:3outMod3} on the uniform distribution with probability $\frac{1}{3} + \varepsilon$. 
By the worst-case to average-case reduction in \Cref{lem:mod3selfreduce}, we get an $\AC^0[2]/\rpoly$ circuit (of similar size and depth) that succeeds with probability $\frac{1}{3} + \varepsilon$ on every input, and outputs each wrong answer with probability $\frac{1}{3} - \frac{\varepsilon}{2}$. 

Construct a circuit for the Mod 3 Problem (\Cref{problem:mod3}) such that on input $x$, 
\begin{itemize}
\item with probability $\frac{1}{4}$, it outputs $0$, 
\item with probability $\frac{3}{4}$, it outputs $0$ if $A(x) = 0$ and $1$ otherwise. 
\end{itemize}
If $|x| \bmod 3 = 0$ then the circuit outputs $0$ w.p. $\frac{1}{4} + \frac{3}{4}(\frac{1}{3} + \varepsilon) = \frac{1}{2} + \frac{3}{4} \varepsilon$. If $|x| \bmod 3 \neq 0$ then the circuit outputs $1$ w.p. $\frac{3}{4}(\frac{1}{3} + \varepsilon + \frac{1}{3} - \frac{\varepsilon}{2}) = \frac{1}{2} + \frac{3}{8} \varepsilon$. In other words, the circuit solves the Mod 3 Problem with probability at least $\frac{1}{2} + \frac{3}{8} \varepsilon$ on arbitrary inputs. By \Cref{lem:avgmod3}, this implies $\varepsilon = \frac{1}{n^{\Omega(1)}}$. 
\end{proof}

From this we get the following corollary.

\begin{cor} \label{cor:randomOutputson3Mod3}
Let $C$ be an $\AC^0[2]/\rpoly$ circuit outputting a trit. 
For any fixed $x\in \{0,1,2\}^n$, 
we denote by $C(x)$ the random variable giving the output of the randomized circuit $C$ on input $x$.
Then, for all $i \in \{0,1,2\}$
\begin{align}
    \frac{1}{3} - n^{-\Omega(1)} \leq \Pr_{x\in \{0,1,2\}^n}\Big[ C(x) - |x| \equiv i \pmod{3} \Big] \leq \frac{1}{3} + n^{-\Omega(1)}.
\end{align}
\end{cor}

\begin{proof}
Since 
\begin{align}
    \sum_{i=0}^2 \Pr\Big[ C(x) - |x| \equiv i \pmod{3} \Big] = 1,
\end{align} 
it suffices to show 
\begin{align}
    \Pr\Big[ C(x) - |x| \equiv i \pmod{3} \Big] \leq \frac{1}{3} + n^{-\Omega(1)}
\end{align} 
for $i\in \{0,1,2\}$.
For $i = 0$ this claim is exactly the statement of  \Cref{lem:3outMod3}. For $i \in \{1,2\}$ we note a circuit satisfying 
\begin{align}
\Pr\Big[ C(x) - |x| \equiv i \pmod{3} \Big] \geq \frac{1}{3} + n^{-o(1)}
\end{align} can be converted to one violating \Cref{lem:3outMod3} by subtracting $i$ from every output. 
\end{proof}

We now need an analog of Vazirani's XOR Lemma for finite groups~{\cite[Lemma~4.2]{rao2007exposition}}.

\begin{lem}[XOR lemma for finite abelian groups]\label{lem:rao} Let $X$ be a distribution on a finite abelian group G such that $\abs{\mathbb{E}\left[\psi(X)\right]} \leq \epsilon$ for every non-trivial character $\psi$. Then $X$ is $\epsilon\sqrt{\abs{G}}$ close (in statistical distance) to the uniform distribution over $G$.
\end{lem}

We now consider the parallel version of the previous problem and show that it is hard.

\begin{problem}[$k$-Parallel 3 Output Mod 3] \label{problem:rep3outMod3}
Given inputs $x_1,\ldots, x_k  \in \{0,1,2\}^n$ for all $i \in [k]$,  output a vector $\vec{y} \in \{0,1,2\}^k$ such that 
\begin{align}
    y_i \equiv \abs{x_i} \pmod{3}
\end{align}
for at least a $\frac{1}{3} + 0.01$ fraction of the $i \in [k]$.
\end{problem}

\begin{thm}
There exists a $k \in \Theta(\log n)$ for which any $\AC^0[2]/\rpoly$ circuit solves the $k$-Parallel 3 Output Mod 3 Problem (\Cref{problem:rep3outMod3}) with probability at most $n^{-\Omega(1)}$.
\end{thm}

\begin{proof}
Our proof is be similar to that of \Cref{thm:RPBP}: we will first prove that solving the sum of any subset of subgames is hard, and then apply \Cref{lem:rao}
to deduce that winning more than $1/3+0.01$ fraction of the games is hard.

Let $C$ be an $\AC^0[2]/\rpoly$ circuit trying to solve \Cref{problem:rep3outMod3}.
Let $y_1, \ldots, y_k$ be its $k$ output trits.
We consider the distribution $X$ over $k$ trits defined by \begin{equation}\bigotimes_{i=1}^k\abs{x_i} - y_i \pmod{3}\end{equation}
for a uniform input $x\in \{0,1,2\}^k$.
We shall show that the distribution $X$ is close to the uniform distribution over $\{0,1,2\}^k$.

Let $\chi_a$ be a non-trivial character of $\mathbb{F}_3^k$. That is $\chi_a(z) = \omega^{\sum_{i=1}^{k} a_i z_i}$ where $\omega$ is a third root of unity, and $a\in \mathbb{F}_3^k$. To show that $X$ is close to uniform it suffices to show that the expectation of $\chi_a(X)$ is small for all non-zero vectors of coefficients $a_1, \ldots, a_k\in \{0,1,2\}$.

Given $a\in \{0,1,2\}^k$ let $S$ be the support of $a$, i.e., the set of indices on which $a_i \neq 0$.
Given strings $x_1, \ldots x_k \in \{0,1,2\}^n$ finding trits $y_1, \ldots y_k$ satisfying 
\begin{align}
    \sum_{i \in S} a_i |x_i| \equiv\sum_{i\in S} a_i y_i \pmod{3}
\end{align}
is at least as hard as solving the 3 Output Mod 3 Problem on the concatenated input $(a_i x_i)_{i\in S}$.
This is true since any circuit solving the former can be converted into a circuit solving the latter by adding an depth-2 circuit with $\exp(|S|) \le \exp(k) \le \poly(n)$ gates that adds the $|S|$ trits $a_i y_i$ modulo 3.
However, the concatenated input $(a_i x_i)_{i\in S}$ is a uniform vector in $\{0,1,2\}^{n|S|}$.
\Cref{cor:randomOutputson3Mod3} then gives that
\begin{align}
   \sum_{i \in S} a_i (|x_i|  - y_i) \pmod{3}
\end{align} has an $\ell_1$ distance at most $n^{-\Omega(1)}$ from the uniform distribution over $\{0,1,2\}$.
Thus, $|\E[\chi_a(X)]| \le n^{\Omega(1)}$.

Applying \Cref{lem:rao} with $G = \mathbb{Z}_3^k$ and $X$ the distribution of the random variables $\bigotimes_{i=1}^k\abs{x_i} - y_i \pmod{3}$ gives that $X$ has $\ell_1$ distance 
\begin{align}
     n^{-\Omega(1)} \cdot \sqrt{3^{k}}
\end{align}
from the uniform distribution on $\{0,1,2\}^k$. 

To finish the proof, we note that by  Chernoff's bound, the probability of drawing a string from the uniform distribution over $\{0,1,2\}^k$ with more than a $\frac{1}{3} + 0.01$ fraction of its outputs $0$ is bounded by $\exp(-\Omega(k))$.
Then we see the probability of drawing a string from $X$ with more than $\frac{1}{3} + 0.01$ fraction of zeros is bounded above by
\begin{align}
    n^{-\Omega(1)} \cdot \sqrt{3^{k}} + e^{-\Omega(k)} = n^{-\Omega(1) + k/(2\log_3(n))} + e^{-\Omega(k)}.
\end{align}
To complete the proof, we note there is some $k \in \Theta(\log n)$ for which the above sum is bounded above by $n^{-\Omega(1)}$.
\end{proof}

\begin{thm}
There exists a $k \in \Theta(\log n)$ for which an $\AC^{0}[2]/\rpoly$ circuit succeeds on the $k$-Parallel Parity Bending Problem with probability at most $n^{-\Omega(1)}$. 
\end{thm}

\begin{proof}
We show a circuit solving the Parallel Parity Bending Problem can be reduced to one solving \Cref{problem:rep3outMod3} with success probability close to $\frac{1}{2}$ the original success probability.

This reduction is straightforward : given a solution $y_1, \ldots y_k$ to \Cref{problem:RPBP} we convert to a solution $y_1', \ldots, y_k'$ to \Cref{problem:rep3outMod3} by setting $y_i' = 0$ if $\abs{y_i} = 0 \pmod{2}$, and $y_i'$ equal to a random choice of 1 or 2 if $\abs{y_i} = 1 \pmod{2}$. The expected number of successes is at least half the number of successes in the original instance, since a success on input $i$ with $|x_i| \equiv0 \pmod{3}$ remains a success with probability 1, and a success on an input $i$ with $|x_i| \pmod{3} \in \{1,2\} $ remains a success with probability $\frac{1}{2}$. Concentrating around this value completes the proof.
\end{proof}

This establishes the classical lower bound in \Cref{thm:RPBP}.

\section*{Acknowledgements}
We would like to thank Nicolas Delfosse, Aram Harrow, Anna G{\'a}l, Anand Natarajan, Benjamin Rossman, and Emanuele Viola for very helpful discussions.

\appendix

\section{Proof of the multi-switching lemma (\texorpdfstring{\Cref{cor:Rossman}}{Lemma \ref*{cor:Rossman}})}
\label{app:Rossman}

We follow the notation that was set up by Rossman in~\cite{Rossman2017} and define the following classes of Boolean functions. 
One difference is that we consider functions with multiple outputs $\{f:\B^n \to \B^m\}$, instead of $\{f:\B^n \to \B\}$.
\begin{itemize}
\item $\DT(k)$ is the class of depth-$k$ decision trees with a single output bit.
\item $\CKT(d; s_1,s_2, \ldots, s_d)$ denotes the class of depth-$d$ $\AC^0$ circuits with $s_i$ nodes at height $i$ for all $i \in \{1,\ldots, d\}$. Note that these circuits compute functions with $s_d$ many output bits.
\item $\CKT(d; s_1,s_2, \ldots, s_d)  \circ \DT(k)$ is the class of circuits in $\CKT(d; s_1,s_2,\ldots, s_d)$ whose inputs are labeled by decision trees in $\DT(k)$.
\item $\DT(t) \circ \CKT(d; s_1,s_2,\ldots, s_d) \circ \DT(k)$ is the class of depth-$t$ decision trees, whose leaves are labeled by elements of $\CKT(d; s_1,\ldots, s_d) \circ \DT(k)$. (Note that these are functions with $s_d$ output bits)
\item $\DT(k)^{m}$ is the class of $m$-tuples of depth-$k$ decision trees. That is a function $F \in \DT(k)^{m}$ is a tuple of $m$ functions $F = (f_1, \ldots, f_m)$ where each $f_i\in \DT(k)$. 
\item $\DT(t) \circ \DT(k)^m$ is the class of depth-$t$ decision trees, whose leaves are labeled by $m$-tuples of depth-$k$ decision trees, one per output bit. 
\end{itemize}

Recall that $\DT(t) \circ \DT(k)^m$ is the class of functions mapping $\B^n$ to $\B^m$ that can be evaluated by adaptively querying at most $t$ coordinates globally, after which each of the $m$ output bits can be evaluated by making at most $k$ additional adaptive queries. Note that while the first $t$ queries are global, the last $k$ queries could be different for each output bit. 

The next lemma shows that under random restriction with high probability objects of the form $\DT(\cdot) \circ \CKT(d; \ldots) \circ \DT(\cdot)$ reduce to objects of the form $\DT(\cdot) \circ \CKT(d-1; \ldots) \circ \DT(\cdot)$, where the depth of the circuit reduces by one. Applying this lemma for $d$ iterations would reduce the depth of the circuit to $1$.

\begin{lem}[\protect{\cite[Lemma~24]{Rossman2017}}] \label{lemma:Rossman-Induction}
Let $d, t, \ell, k, s_1, \ldots, s_d \in \mathbb{N}$, $d \ge 2$, $p \in (0,1)$.
If $\ell \ge \log(s_1) + 1$ and $f \in \DT(t-1) \circ \CKT(d; s_1, \ldots, s_d) \circ \DT(k)$, then
\begin{equation}
\Pr_{\rho \sim \Rp} [f|_\rho \notin
\DT(t-1)\circ \CKT(d-1; s_2, \ldots, s_d) \circ \DT(\ell)]
 \le s_1 (200 pk)^{t/2}\;.
\end{equation}
\end{lem}

\begin{lem}[Multi-Switching Lemma \cite{Hastad14}, 
for this formulation see~\protect{\cite[Theorem~40]{Tal17}}] \label{lemma:Multi-Switching}
Let $m, k, q, t \in \mathbb{N}$, $p \in (0,1)$.
If $f \in \CKT(1; m) \circ \DT(k)$, then
\begin{equation}
\Pr_{\rho \sim \mathbf{R}_p}
[f|_\rho \notin
\DT(t-1)\circ \DT(q-1)^{m}]
 \le m^{1+t/q} \cdot (25 pk)^{t}\;.    
\end{equation}
\end{lem}

\begin{lem}[Multiple output $\AC^0$ circuits under random restrictions]
\label{lemma:technical}
Let $d, t, q, k, s_1, \ldots, s_{d-1}  \in \mathbb{N}$, $p_1, \ldots, p_d \in (0,1)$.
Let $f \in \CKT(d; s_1, \ldots, s_{d-1},m)$ with $n$ inputs and $m$ outputs.
Let $s = s_1 + \ldots + s_{d-1}+m$.
Let $p = p_1 \cdot p_2 \cdots p_d$.
Then
\begin{align}
\Pr_{\rho \sim \Rp}[f|_{\rho} \notin \DT(2t-2) & \circ \DT(q-1)^m] \nonumber \\ 
&\le 
s_1 \cdot O(p_1)^{t/2}
+
\left(\sum_{i=2}^{d-1} 
s_i \cdot O(p_i \log s)^{t/2}\right) + 
m^{1+t/q} \cdot O(p_d \log s)^{t}.
\end{align}
\end{lem}
\begin{proof}
Let $s_d = m$ for ease of notation.
Let $\ell = \lceil{\log(s)\rceil} +1$.
We think of $\rho \sim \Rp$ as a composition of $d$ random restrictions $\rho_1 \circ \ldots \circ \rho_d$ where each $\rho_i \sim \Rest_{p_i}$.
For $i=1, \ldots, d-1$ we let $\cE_i$ be the event defined by 
\begin{equation}
	\cE_i 
	\;\;\Leftrightarrow \;\;
	f|_{\rho_1 \circ \cdots \circ \rho_i}  \in  \DT(t-1) \circ \CKT(d-i; s_{i+1}, \ldots, s_d) \circ \DT(\ell),
\end{equation}
	and denote by $\cE_d$ the event that 
\begin{equation}
\cE_d \;\;\Leftrightarrow\;\; f|_{\rho_1 \circ \ldots \circ \rho_d}  \in  \DT(2t-2)  \circ \DT(q-1)^{m}.
\end{equation}
	We shall show that $\cE_1 \wedge \ldots \wedge \cE_d$ happens with high probability.
	We start with the first event $\cE_1$. Since $f \in \CKT(d; s_1, \ldots; s_d)$, it also holds that 
	\begin{equation} f\in \DT(t-1) \circ \CKT(d; s_1, \ldots, s_d) \circ \DT(1).\end{equation} Thus, we may apply \Cref{lemma:Rossman-Induction} with $k = 1$ and get
		\begin{equation}
	\Pr[\neg \cE_1] \le s_1 \cdot O(p_1 )^{t/2}
	.\end{equation}
	For $i = 2, \ldots, d-1$,
	using \Cref{lemma:Rossman-Induction} again,  we have
	\begin{equation}
	\Pr[\neg \cE_i | \cE_1 \wedge \ldots \wedge \cE_{i-1}] \le s_i \cdot O(p_i \cdot \ell )^{t}.
	\end{equation}
	
	We are left with the last event -- showing that $\Pr[\neg \cE_d | \cE_1 \wedge \ldots \wedge \cE_{d-1}]$ is small.
	We condition on $\rho_1, \ldots, \rho_{d-1}$ satisfying $\cE_1 \wedge \ldots \wedge \cE_{d-1}$, and  denote by $g = f|_{\rho_1 \circ \ldots \circ \rho_{d-1}}$.
	Under this conditioning, we have that 
	\begin{equation}g \in \DT(t-1) \circ \CKT(1; m) \circ \DT(\ell).\end{equation}
	For each leaf $\lambda$ of  the partial decision tree of depth at most $t-1$ for $g$, denote by $g_{\lambda}$ the function $g$ restricted by the partial assignment made along the path to $\lambda$. Note that for each leaf $\lambda$, the function $g_{\lambda}$ is in $\CKT(1;m) \circ \DT(\ell)$.
	By \Cref{lemma:Multi-Switching}, for each leaf $\lambda$, 
	\begin{equation}\Pr[g_{\lambda}|_{\rho_d} \notin \DT(t-1) \circ \DT(q-1)^{m}] \le 
m^{1+t/q} \cdot O(p_d \cdot \ell)^{t}.
\end{equation}
Since there are at most $2^{t-1}$ leaves, by union bound, the probability that there exists a leaf $\lambda$ for which the switching does not hold is at most 
\begin{equation}
2^{t-1} \cdot m^{1+t/q} \cdot O(p_d \cdot \ell)^{t}\;.
\end{equation}
In the complement event, the function $g|_{\rho_d} = f|_{\rho_1 \circ \ldots \circ \rho_d}$ is in $\DT(t-1)\circ \DT(t-1) \circ  \DT(q-1)^m$ and thus $\cE_d$ holds. Thus, we showed that 
\begin{equation}
\Pr[\neg \cE_d | \cE_1 \wedge \ldots \wedge \cE_{d-1}] \le 
2^{t-1} \cdot m^{1+t/q} \cdot O(p_d \cdot \ell)^{t} \le m^{1+t/q} \cdot O(p_d \cdot \ell)^{t}.
\end{equation}
Overall, we got that the probability that one of $\cE_1, \ldots,\cE_d$ does not hold is
\begin{align}
\Pr[\neg \cE_1 \vee \ldots \vee \neg \cE_d)]  &= 
\sum_{i=1}^{d} \Pr[\neg \cE_i | \cE_1 \wedge \ldots \cE_{i-1}]
\\
&\le s_1 \cdot O(p_1)^{t/2} + 
\left(\sum_{i=2}^{d-1} s_i \cdot O(p_i \cdot \ell)^{t/2}\right) 
+
m^{1+t/q} \cdot O(p_d \cdot \ell)^{t}\;,
\end{align}
as promised.
\end{proof}

We are now ready to restate and prove \Cref{cor:Rossman}.
(Note that the formulation here is slightly stronger than what we used in \Cref{sec:AC}, as we replace $2t$ and $q$ with $2t-2$ and $q-1$, respectively.)
\begin{lem}
	Let $f:\B^n \to \B^m$ be an $\AC^0$ circuit of size $s$, depth $d$.
Let $p = 1/ (m^{1/q} \cdot O(\log s)^{d-1})$.
Then 
\begin{equation}
\Pr_{\rho \sim \Rp}[f|_{\rho} \notin \DT(2t-2) \circ \DT(q-1)^m] \le 
s \cdot 2^{-t}.
\end{equation}
\end{lem}

\begin{proof}
	We apply \Cref{lemma:technical} with $p_1 = 1/O(1)$ and $p_2 = \ldots = p_{d-1} = 1/O(\log s)$ and $p_d = 1/ O(m^{1/q} \cdot \log s)$.
\end{proof}

\bibliographystyle{alphaurl}
\phantomsection\addcontentsline{toc}{section}{References} 
\bibliography{ref}{}

\end{document}

%% file: shortdepth.tex
\begin{tikzpicture}[scale=1.000000,x=1pt,y=1pt]
\filldraw[color=white] (0.000000, -7.500000) rectangle (96.000000, 82.500000);
\draw[color=black] (0.000000,75.000000) -- (96.000000,75.000000);
\draw[color=black] (0.000000,75.000000) node[left] {$x_1$};
\draw[color=black] (0.000000,60.000000) -- (96.000000,60.000000);
\draw[color=black] (0.000000,60.000000) node[left] {$x_2$};
\draw[color=black] (0.000000,45.000000) -- (96.000000,45.000000);
\draw[color=black] (0.000000,45.000000) node[left] {$x_3$};
\draw[color=black] (0.000000,30.000000) -- (84.000000,30.000000);
\draw[color=black] (84.000000,29.500000) -- (96.000000,29.500000);
\draw[color=black] (84.000000,30.500000) -- (96.000000,30.500000);
\draw[color=black] (0.000000,15.000000) -- (84.000000,15.000000);
\draw[color=black] (84.000000,14.500000) -- (96.000000,14.500000);
\draw[color=black] (84.000000,15.500000) -- (96.000000,15.500000);
\draw[color=black] (0.000000,0.000000) -- (84.000000,0.000000);
\draw[color=black] (84.000000,-0.500000) -- (96.000000,-0.500000);
\draw[color=black] (84.000000,0.500000) -- (96.000000,0.500000);
\filldraw[color=white,fill=white] (0.000000,-3.750000) rectangle (-4.000000,33.750000);
\draw[decorate,decoration={brace,amplitude = 4.000000pt},very thick] (0.000000,-3.750000) -- (0.000000,33.750000);
\draw[color=black] (-4.000000,15.000000) node[left] {$\ket{\Cat_3}$};
\draw (12.000000,75.000000) -- (12.000000,30.000000);
\begin{scope}
\draw[fill=white] (12.000000, 30.000000) +(-45.000000:8.485281pt and 8.485281pt) -- +(45.000000:8.485281pt and 8.485281pt) -- +(135.000000:8.485281pt and 8.485281pt) -- +(225.000000:8.485281pt and 8.485281pt) -- cycle;
\clip (12.000000, 30.000000) +(-45.000000:8.485281pt and 8.485281pt) -- +(45.000000:8.485281pt and 8.485281pt) -- +(135.000000:8.485281pt and 8.485281pt) -- +(225.000000:8.485281pt and 8.485281pt) -- cycle;
\draw (12.000000, 30.000000) node {$S$};
\end{scope}
\filldraw (12.000000, 75.000000) circle(1.500000pt);
\draw (24.000000,60.000000) -- (24.000000,15.000000);
\begin{scope}
\draw[fill=white] (24.000000, 15.000000) +(-45.000000:8.485281pt and 8.485281pt) -- +(45.000000:8.485281pt and 8.485281pt) -- +(135.000000:8.485281pt and 8.485281pt) -- +(225.000000:8.485281pt and 8.485281pt) -- cycle;
\clip (24.000000, 15.000000) +(-45.000000:8.485281pt and 8.485281pt) -- +(45.000000:8.485281pt and 8.485281pt) -- +(135.000000:8.485281pt and 8.485281pt) -- +(225.000000:8.485281pt and 8.485281pt) -- cycle;
\draw (24.000000, 15.000000) node {$S$};
\end{scope}
\filldraw (24.000000, 60.000000) circle(1.500000pt);
\draw (36.000000,45.000000) -- (36.000000,0.000000);
\begin{scope}
\draw[fill=white] (36.000000, 0.000000) +(-45.000000:8.485281pt and 8.485281pt) -- +(45.000000:8.485281pt and 8.485281pt) -- +(135.000000:8.485281pt and 8.485281pt) -- +(225.000000:8.485281pt and 8.485281pt) -- cycle;
\clip (36.000000, 0.000000) +(-45.000000:8.485281pt and 8.485281pt) -- +(45.000000:8.485281pt and 8.485281pt) -- +(135.000000:8.485281pt and 8.485281pt) -- +(225.000000:8.485281pt and 8.485281pt) -- cycle;
\draw (36.000000, 0.000000) node {$S$};
\end{scope}
\filldraw (36.000000, 45.000000) circle(1.500000pt);
\begin{scope}
\draw[fill=white] (60.000000, 30.000000) +(-45.000000:8.485281pt and 8.485281pt) -- +(45.000000:8.485281pt and 8.485281pt) -- +(135.000000:8.485281pt and 8.485281pt) -- +(225.000000:8.485281pt and 8.485281pt) -- cycle;
\clip (60.000000, 30.000000) +(-45.000000:8.485281pt and 8.485281pt) -- +(45.000000:8.485281pt and 8.485281pt) -- +(135.000000:8.485281pt and 8.485281pt) -- +(225.000000:8.485281pt and 8.485281pt) -- cycle;
\draw (60.000000, 30.000000) node {$H$};
\end{scope}
\begin{scope}
\draw[fill=white] (60.000000, 15.000000) +(-45.000000:8.485281pt and 8.485281pt) -- +(45.000000:8.485281pt and 8.485281pt) -- +(135.000000:8.485281pt and 8.485281pt) -- +(225.000000:8.485281pt and 8.485281pt) -- cycle;
\clip (60.000000, 15.000000) +(-45.000000:8.485281pt and 8.485281pt) -- +(45.000000:8.485281pt and 8.485281pt) -- +(135.000000:8.485281pt and 8.485281pt) -- +(225.000000:8.485281pt and 8.485281pt) -- cycle;
\draw (60.000000, 15.000000) node {$H$};
\end{scope}
\begin{scope}
\draw[fill=white] (60.000000, 0.000000) +(-45.000000:8.485281pt and 8.485281pt) -- +(45.000000:8.485281pt and 8.485281pt) -- +(135.000000:8.485281pt and 8.485281pt) -- +(225.000000:8.485281pt and 8.485281pt) -- cycle;
\clip (60.000000, 0.000000) +(-45.000000:8.485281pt and 8.485281pt) -- +(45.000000:8.485281pt and 8.485281pt) -- +(135.000000:8.485281pt and 8.485281pt) -- +(225.000000:8.485281pt and 8.485281pt) -- cycle;
\draw (60.000000, 0.000000) node {$H$};
\end{scope}
\draw[fill=white] (78.000000, 24.000000) rectangle (90.000000, 36.000000);
\draw[very thin] (84.000000, 30.600000) arc (90:150:6.000000pt);
\draw[very thin] (84.000000, 30.600000) arc (90:30:6.000000pt);
\draw[->,>=stealth] (84.000000, 24.600000) -- +(80:10.392305pt);
\draw[fill=white] (78.000000, 9.000000) rectangle (90.000000, 21.000000);
\draw[very thin] (84.000000, 15.600000) arc (90:150:6.000000pt);
\draw[very thin] (84.000000, 15.600000) arc (90:30:6.000000pt);
\draw[->,>=stealth] (84.000000, 9.600000) -- +(80:10.392305pt);
\draw[fill=white] (78.000000, -6.000000) rectangle (90.000000, 6.000000);
\draw[very thin] (84.000000, 0.600000) arc (90:150:6.000000pt);
\draw[very thin] (84.000000, 0.600000) arc (90:30:6.000000pt);
\draw[->,>=stealth] (84.000000, -5.400000) -- +(80:10.392305pt);
\draw[color=black] (96.000000,75.000000) node[right] {$x_1$};
\draw[color=black] (96.000000,60.000000) node[right] {$x_2$};
\draw[color=black] (96.000000,45.000000) node[right] {$x_3$};
\draw[color=black] (96.000000,30.000000) node[right] {$y_1$};
\draw[color=black] (96.000000,15.000000) node[right] {$y_2$};
\draw[color=black] (96.000000,0.000000) node[right] {$y_3$};
\end{tikzpicture}